\documentclass[acmsmall,nonacm,screen]{acmart}  

\usepackage{algorithm}
\usepackage[noend]{algpseudocode}
\usepackage{alltt}
\usepackage{booktabs}
\usepackage[commandnameprefix=ifneeded,final]{changes}  
\usepackage{enumitem}
\usepackage{multirow}
\usepackage[autolanguage]{numprint} \npthousandsep{,}
\usepackage{soul}
\usepackage{subcaption}
\usepackage{textcomp}

\graphicspath{{figs/}{./}}

\renewcommand{\Pr}[1]{\ensuremath{\mathrm{Pr}\left[#1\right]}}
\newcommand{\E}[1]{\ensuremath{\mathrm{E}\left[#1\right]}}

\newcommand{\punish}{\pi}
\newcommand{\adamancy}{\alpha}
\newcommand{\tolerate}[1]{\tau_{#1}}
\newcommand{\severity}[1]{\psi_{#1}}
\newcommand{\surveil}[1]{\nu_{#1}}
\newcommand{\desire}[1]{\delta_{#1}}
\newcommand{\boldness}[1]{\beta_{#1}}
\newcommand{\action}[2]{a_{#1,#2}}
\newcommand{\estimatetol}[1]{\hat{\tau}_{#1}}
\newcommand{\estimatesev}[1]{\hat{\psi}_{#1}}
\newcommand{\estimatesur}[1]{\hat{\nu}_{#1}}

\makeatletter
\algrenewcommand\ALG@beginalgorithmic{\small}
\algrenewcommand\alglinenumber[1]{\footnotesize #1:}
\makeatother

\makeatletter
\newif\ifanon
\@ifclasswith{acmart}{anonymous}{\anontrue}{\anonfalse}
\makeatother

\newif\ifcomment
\commenttrue    

\newif\iffigabbrv
\figabbrvtrue    
\newcommand{\figtext}{\iffigabbrv Fig.\else Figure\fi}
\newcommand{\figstext}{\iffigabbrv Figs.\else Figures\fi}

\newif\ifeqabbrv
\eqabbrvtrue    
\newcommand{\eqtext}{\ifeqabbrv Eq.\else Equation\fi}
\newcommand{\eqstext}{\ifeqabbrv Eqs.\else Equations\fi}

\title{Strategic Analysis of Dissent and Self-Censorship}

\author{Joshua J. Daymude}
\orcid{https://orcid.org/0000-0001-7294-5626}
\affiliation{%
    \institution{Arizona State University}
    \department{Biodesign Center for Biocomputing, Security and Society}
    \department{School of Computing and Augmented Intelligence}
    \streetaddress{727 E. Tyler St.}
    \city{Tempe}
    \state{AZ}
    \postcode{85281}
    \country{USA}}
\email{jdaymude@asu.edu}

\author{Robert Axelrod}
\orcid{https://orcid.org/0000-0003-4758-6590}
\affiliation{%
    \institution{University of Michigan}
    \department{School of Public Policy}
    \city{Ann Arbor}
    \state{MI}
    \postcode{48109}
    \country{USA}}
\email{axe@umich.edu}

\author{Stephanie Forrest}
\orcid{https://orcid.org/0000-0002-5904-1646}
\affiliation{%
    \institution{Arizona State University}
    \department{Biodesign Center for Biocomputing, Security and Society}
    \department{School of Computing and Augmented Intelligence}
    \streetaddress{727 E. Tyler St.}
    \city{Tempe}
    \state{AZ}
    \postcode{85281}
    \country{USA}}
\affiliation{%
    \institution{Santa Fe Institute}
    \city{Santa Fe}
    \state{NM}
    \postcode{87501}
    \country{USA}}
\email{steph@asu.edu}

\begin{abstract}
    Expressions of dissent against authority are an important feature of most societies, and efforts to suppress such expressions are common.
    Modern digital communications, social media, and Internet surveillance and censorship technologies are changing the landscape of public speech and dissent.
    Especially in authoritarian settings, individuals must assess the risk of voicing their true opinions or choose \textit{self-censorship}, voluntarily moderating their behavior to comply with authority.
    We present a model in which individuals strategically manage the tradeoff between expressing dissent and avoiding punishment through self-censorship while an authority adapts its policies to minimize both total expressed dissent and punishment costs.
    We study the model analytically and in simulation to derive conditions separating \textit{defiant individuals} who express their desired dissent in spite of punishment from \textit{self-censoring individuals} who fully or partially limit their expression.
    We find that for any population, there exists an authority policy that leads to total self-censorship.
    However, the probability and time for an initially moderate, locally-adaptive authority to suppress dissent depend critically on the population's willingness to withstand punishment early on, which can deter the authority from adopting more extreme policies.
\end{abstract}

\begin{document}

\maketitle

\section{Introduction}

Dissent against authority is ubiquitous throughout human history and culture~\cite{Hsiao2020-versobook}, as are authorities' efforts to maintain power, stability, and social order.
But modern Internet infrastructure, encrypted communication protocols, social media platforms, and automated surveillance technologies pose new opportunities and threats for potential dissidents~\cite{Lokot2021-protestsquare,Feldstein2021-risedigital}.
Fueled by the last decade of advances in computer vision~\cite{Kalluri2025-computervisionresearch}, facial recognition tools in particular have been used around the world to track and arrest protesters, inducing a chilling effect~\cite{Mozur2019-hongkong,Vincent2020-nypdused,AmnestyInternational2023-automatedapartheid,Loewenstein2023-palestinelaboratory,Loucaides2024-changingface,Morris2024-whyfacial}.
Internet censorship systems such as China's Great Firewall~\cite{Ensafi2015-analyzinggreat,Roberts2018-censoreddistraction,Fedasiuk2021-buyingsilence} and Russia's ``TSPU'' devices for centralized, nation-wide, real-time deep packet inspection~\cite{Epifanova2021-subjugatingrunet,Xue2022-tspurussia} monitor traffic for site access or material considered objectionable.
The Cyberspace Administration of China explicitly signals its ability to punish what it finds: ``The cyber police are right by your side... You will exercise restraint and rationality when you post and write messages''~\cite{CAC2015-guangdonginternet}.
The U.S.\ values free speech as a constitutional right; however, there are numerous tensions around privacy and acceptable speech, ranging from U.S.\ technology giants' heterogeneous and opaque practices of ``content moderation''~\cite{Gillespie2018-custodiansinternet,MyersWest2018-censoredsuspended,Common2020-fearreaper,Wilson2021-hatespeech,Chowdhury2022-examiningfactors,Nicholas2022-sheddinglight,Toraman2022-blacklivesmatter2020,HumanRightsWatch2023-metabroken,Pierri2023-howdoes,Shahid2023-decolonizingcontent,Kaplan2025-morefree} to the FBI's four million warrantless data queries of American's ``incidentally collected'' private communications in 2020--2023~\cite{ODNI2024-annualstatistical} to the January--February 2025 Executive Orders aiming to ``end federal censorship'' while simultaneously deeming any ``diversity, equity, and inclusion'' language and programs within the federal government illegal~\cite{Trump2025-restoringfreedom,Trump2025-endingradical,Trump2025-endingillegal}.

In these settings, individuals must make careful choices about if, when, and how to express dissent: Do they voice their true opinions in spite of the potential consequences, or do they engage in \textit{self-censorship}~\cite{Loury1994-selfcensorshippublic,Horton2011-selfcensorship,Bar-Tal2017-selfcensorshipsociopoliticalpsychological,Shen2021-searchselfcensorship}---voluntarily moderating their behavior to avoid punishment by the authority---thus granting the authority an implicit form of control?
These considerations point to a set of strategic questions about self-censorship: when it emerges, when it is resisted, its relationship to different punishment regimes, and how interactions between the population and the authority change these dynamics over time.

This work resides in the gap between the intuitive sense and formative theory of civil resistance as a distributed, strategic, and uncertain process~\cite{Ackerman2008-strategicdimensions,Schelling1971-questionscivilian,Sharp1973-politicsnonviolent} and the population-level, comparative, statistical methods typically used to study them~\cite{Chenoweth2011-whycivil,Chenoweth2021-civilresistance}.
We present a model comprising a single authority and a population of individuals.
The individuals strategically decide what level of dissent to express based on the strength of their inherent desire to dissent, their boldness, and the perceived risk of punishment by the authority.
We derive analytical results for the individuals' optimal actions under two forms of authoritarian punishment, uniform and proportional, and compare the model's predictions to contemporary examples.
Finally, we investigate how a population's desired dissent and boldness distributions influence a locally-adaptive authority's policies, and how that in turn affects population behavior.

\section{Model}

We model a single authority $A$ (e.g., a government, corporation, or social media platform) governing a population of $n$ individuals $\mathcal{I} = \{I_1, \ldots, I_n\}$ (e.g., citizens, employees, or users, respectively).
Each individual $I_i$ has a level of \textit{desired dissent} denoted by $\desire{i} \in [0, 1]$, where $\desire{i} = 0$ is no desire and $\desire{i} = 1$ is maximal.  
In each round $r = 1, 2, \ldots$, each individual $I_i$ takes an \textit{action} with a level of expressed dissent denoted by $\action{i}{r} \in [0, \desire{i}]$.
The authority then observes each action with probability
\begin{equation} \label{eq:observe}
    \Pr{A \text{ observes } \action{i}{r}} = \surveil{r} + (1 - \surveil{r}) \cdot \action{i}{r},
\end{equation}
where $\surveil{r} \in [0, 1]$ is the authority's \textit{surveillance} in round $r$.
This reflects an assumption that the more dissenting an action is, the more likely the authority is to observe it.
The authority then punishes each individual $I_i$ according to its \textit{punishment function} $\punish$ based on the actions it observes, its \textit{tolerance} $\tolerate{r} \in [0,1]$, and its \textit{severity} $\severity{r} > 0$ in round $r$:
\begin{equation} \label{eq:punish}
    \punish(\action{i}{r}, \tolerate{r}, \severity{r}, \surveil{r}) =
    \left\{\begin{array}{cl}
        0 & \text{if $A$ does not observe $\action{i}{r}$} \text{or $\action{i}{r} \leq \tolerate{r}$}; \\
        \severity{r} & \text{if $A$ observes $\action{i}{r} > \tolerate{r}$} \text{and $\punish$ is \underbar{uniform}}; \\
        \severity{r} \cdot (\action{i}{r} - \tolerate{r}) & \text{if $A$ observes $\action{i}{r} > \tolerate{r}$} \text{and $\punish$ is \underbar{proportional}}.
    \end{array}\right.
\end{equation}
We consider two functional forms for $\punish$: a \textit{uniform} function where the authority punishes all infractions equally and a \textit{proportional} (linear) function where ``the punishment fits the crime.''
For example, a government restricting press access for any reporter it views as antagonistic is uniform punishment, while fining a speeding driver based on their infraction over the speed limit is proportional punishment.
In digital spaces, an example of uniform punishment is Internet shutdowns (i.e., blocking all Internet access regardless of content) while social media companies' content moderation policies often employ proportional punishment, with consequences ranging from content removal to account demonetization to permanent bans depending on the type and frequency of violations~\cite{YouTube2025-youtubecommunity}.
Only actions that exceed the authority's tolerance $\tolerate{r}$ and are observed by the authority are punished.

The authority's utility at the end of round $r$ reflects its preference for (\textit{i}) as little dissent as possible while (\textit{ii}) incurring as little punishment cost as possible: 
\begin{equation} \label{eq:authutility}
    U_{A,r} = -\adamancy \cdot \sum_{i=1}^n \action{i}{r} - \sum_{i=1}^n \punish(\action{i}{r}, \tolerate{r}, \severity{r}, \surveil{r}),
\end{equation}
where $\adamancy > 0$ is the authority's \textit{adamancy}, which reflects the relative importance of expressed dissent compared to the cost of punishing. 
The utility of an individual $I_i$ at the end of round $r$ represents its preference for (\textit{i}) acting as closely as possible to its desire while (\textit{ii}) avoiding punishment:
\begin{equation} \label{eq:indutility}
    U_{i,r} = \boldness{i} \cdot (1 - \desire{i} + \action{i}{r}) - \punish(\action{i}{r}, \tolerate{r}, \severity{r}, \surveil{r}),
\end{equation}
where $\boldness{i} > 0$ is a constant representing how \textit{bold} individual $I_i$ is, i.e., the relative strength of its preference to act according to its desires relative to the need to avoid punishment.

\section{Results}

\subsection{Individuals' Optimal Actions: Compliance, Self-Censorship, and Defiance}

We begin by analyzing and interpreting the optimal (utility-maximizing) action $\action{i}{r}^*$ for an individual $I_i$ in round $r$ when the authority's parameters are fixed.
Beyond its own desired dissent $\desire{i}$ and boldness $\boldness{i}$, we assume that individual $I_i$ knows the authority's punishment function $\punish$ and the values of the authority's tolerance $\tolerate{r}$, severity $\severity{r}$, and surveillance $\surveil{r}$.\footnote{All results in this section also hold if the authority's exact parameter values are replaced by individuals' noisy estimates thereof; see SI Appendix for details.}
Complete derivations of all results in this section are given in the SI Appendix.

\subsubsection{Uniform Punishment}

\begin{figure}[t]
    \centering
    \includegraphics[width=0.38\textwidth]{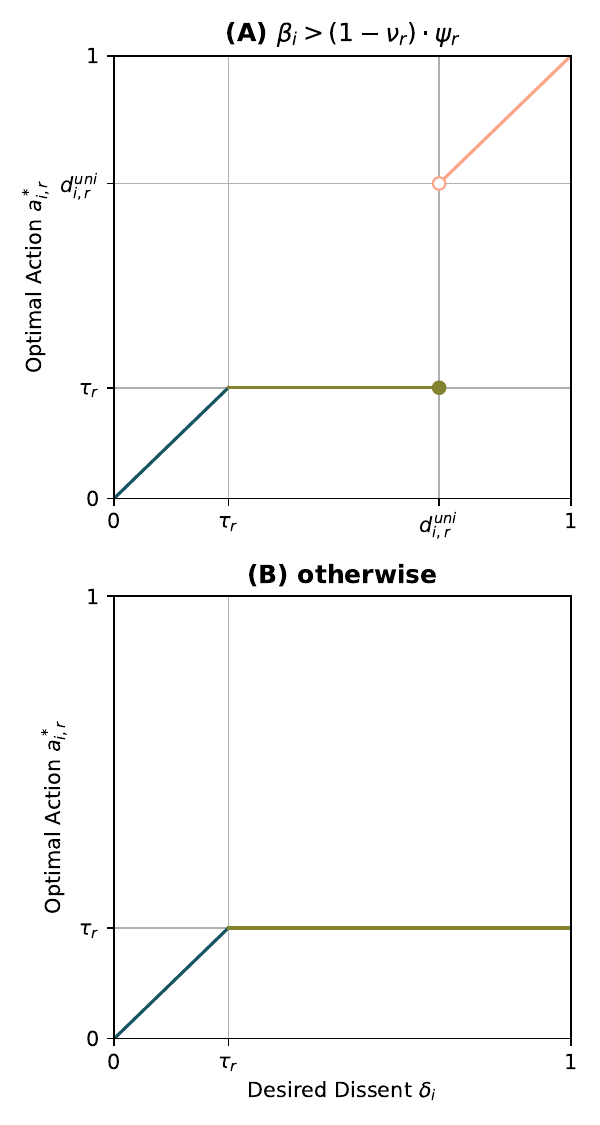}
    \caption{\textit{Individuals' Optimal Actions Under Uniform Punishment.}
    The optimal action of an individual $I_i$ as a function of its desired dissent $\desire{i}$ when $\punish$ is uniform (\eqtext~\ref{eq:optaction:uniform}) and \textbf{(A)} sufficient boldness $\boldness{i} > (1 - \surveil{r}) \cdot \severity{r}$ makes defiance possible or \textbf{(B)} otherwise, when full self-censorship always occurs above $\tolerate{r}$.
    Compliance is shown in teal, self-censorship in olive, and defiance in peach.}
    \label{fig:optaction:uniform}
\end{figure}

When $\punish$ is uniform, the optimal action for an individual $I_i$ is 
\begin{equation} \label{eq:optaction:uniform}
    \action{i}{r}^* = \left\{\begin{array}{cl}
        \desire{i} & \text{if $\desire{i} \leq \tolerate{r}$} \text{or $\big(\boldness{i} > (1 - \surveil{r}) \cdot \severity{r}$ and $\desire{i} > d_{i,r}^{\text{uni}}\big)$}; \\
        \tolerate{r} & \text{otherwise},
    \end{array}\right.
\end{equation}
where
\begin{equation} \label{eq:optaction:dcon}
    d_{i,r}^{\text{uni}} = \frac{\boldness{i} \cdot \tolerate{r} + \surveil{r} \cdot \severity{r}}{\boldness{i} - (1 - \surveil{r}) \cdot \severity{r}}.
\end{equation}
\figtext~\ref{fig:optaction:uniform} shows $\action{i}{r}^*$ as a function of desired dissent $\desire{i}$.

This result categorizes individual responses into three types according to their desired dissents: (\textit{i}) \textit{compliant} individuals whose desired dissents are below the authority's tolerance, enabling them to act as they prefer without punishment; (\textit{ii}) \textit{self-censoring} individuals whose desired dissent exceeds the authority's tolerance but who act at the authority's tolerance level instead, preferring to avoid punishment to expressing dissent; and (\textit{iii}) \textit{defiant} individuals whose boldness and desired dissents are so large that the utility lost by not acting as they prefer outweighs the cost of punishment.
Of these three types, only self-censoring individuals act at lower levels of dissent than their desires.

The most interesting result of this analysis is the distinction between self-censorship and defiance. 
Individuals in both categories desire to act above the authority's tolerance, but analysis of \eqstext~\ref{eq:optaction:uniform}--\ref{eq:optaction:dcon} shows that defiance requires $\boldness{i} > \severity{r}/(1 - \tolerate{r})$; i.e., any defiant individual must be at least as bold as the severity of punishment.

\subsubsection{Proportional Punishment}

\begin{figure}[t]
    \centering
    \includegraphics[width=\textwidth]{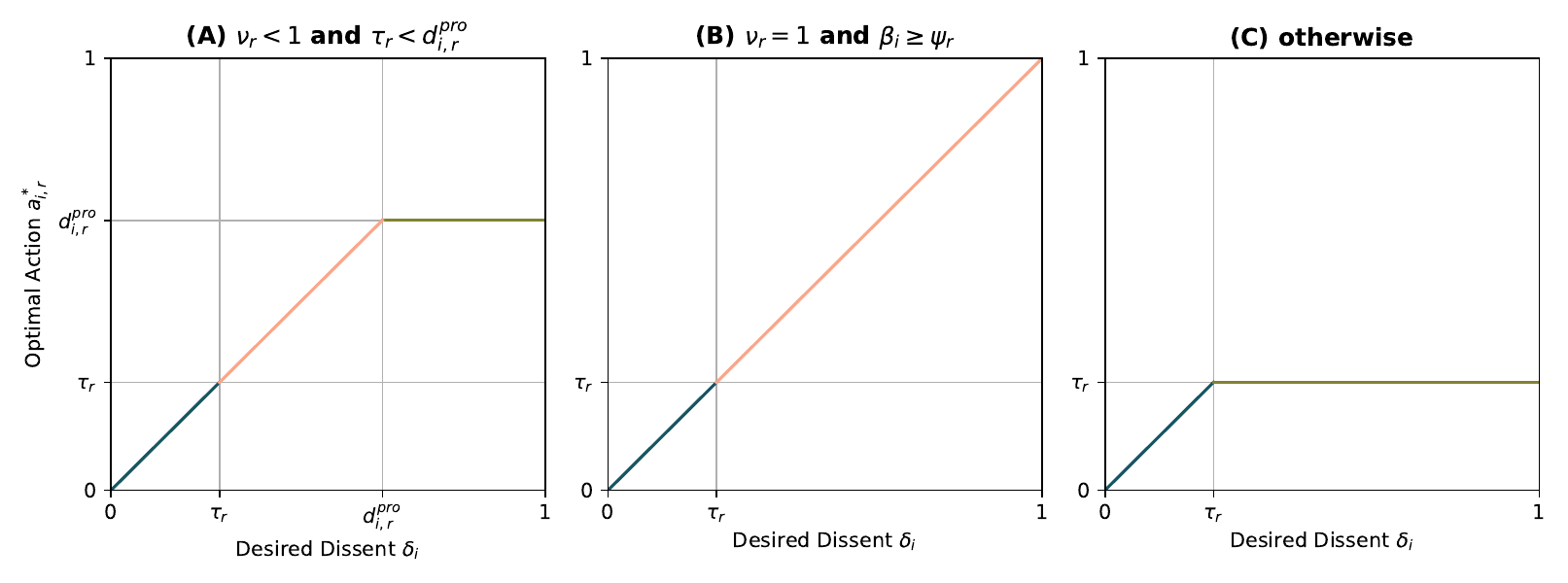}
    \caption{\textit{Individuals' Optimal Actions Under Proportional Punishment.}
    The optimal action of an individual $I_i$ as a function of its desired dissent $\desire{i}$ when $\punish$ is proportional (\eqtext~\ref{eq:optaction:prop}) and \textbf{(A)} the authority's surveillance is imperfect with $\surveil{r} < 1$ and the authority's tolerance $\tolerate{r} < d_{i,r}^{\text{pro}}$ is insufficient to eliminate defiance, \textbf{(B)} the authority's surveillance is perfect with $\surveil{r} = 1$ but the individual's boldness $\boldness{i} \geq \severity{r}$ yields only defiance above $\tolerate{r}$, or \textbf{(C)} otherwise, when full self-censorship always occurs above $\tolerate{r}$.
    Compliance is shown in teal, self-censorship in olive, and defiance in peach.}
    \label{fig:optaction:prop}
\end{figure}

When $\punish$ is proportional to the level of expressed dissent, the optimal action for an individual $I_i$ is
\begin{equation} \label{eq:optaction:prop}
    \action{i}{r}^* = \left\{\begin{array}{cl}
        \desire{i} & \text{if $\desire{i} \leq \tolerate{r}$ or ($\surveil{r} = 1$ and $\boldness{i} \geq \severity{r}$) or ($\surveil{r} < 1$ and $\desire{i} \leq d_{i,r}^{\text{pro}}$)}; \\
        \tolerate{r} & \text{if $\desire{i} > \tolerate{r}$ and (($\surveil{r} = 1$ and $\boldness{i} < \severity{r}$) or ($\surveil{r} < 1$ and $\tolerate{r} \geq d_{i,r}^{\text{pro}}$))}; \\
        d_{i,r}^{\text{pro}} & \text{otherwise},
    \end{array}\right.
\end{equation}
where
\begin{equation} \label{eq:optaction:dpro}
    d_{i,r}^{\text{pro}} = \frac{1}{2}\left(\tolerate{r} + \frac{\boldness{i} - \surveil{r} \cdot \severity{r}}{(1 - \surveil{r}) \cdot \severity{r}}\right).
\end{equation}

As in the case when $\punish$ is uniform, there are \textit{compliant} individuals whose preferred actions are below the authority's tolerance, allowing them to act as desired without punishment.
But surprisingly, in the regime of imperfect surveillance, proportional punishment inverts the parameter regions where which self-censorship and defiance occur (\figtext~\ref{fig:optaction:prop}A).
Here, \textit{defiant} individuals who act as they desire despite potential punishment are not those with the largest desired dissent---as they are when $\punish$ is uniform---but those with intermediate desires.
Intuitively, because both the probability of being observed and the subsequent punishment if observed scale linearly with the infraction, small infractions can be worth the relatively low risk of detection, if it means acting as desired.
However, above the tipping point where the risk of detection and cost of subsequent punishment outweigh desire, individuals become \textit{partially self-censoring}, acting below their desired dissent but still above the authority's tolerance.
The existence and precise location of the tipping point depends critically on the relation between boldness, tolerance, severity, and surveillance.
This partial self-censorship arises only under proportional punishment.

If the authority's surveillance is perfect ($\surveil{r} = 1$), observing every infraction, then an individual's optimal action depends only on the relation between its boldness and the authority's severity: if $\boldness{i} \geq \severity{r}$, then $I_i$ is defiant, acting exactly as it desires regardless of punishment (\figtext~\ref{fig:optaction:prop}B); otherwise, it self-censors to avoid punishment (\figtext~\ref{fig:optaction:prop}C).
Thus, perfect surveillance removes the optimal action's dependence on desired dissent: defiant individuals are sufficiently bold; fully self-censoring individuals are not.

\subsection{Contemporary Illustrations of the Model}

\begin{figure}[t]
    \centering
    \includegraphics[width=\textwidth]{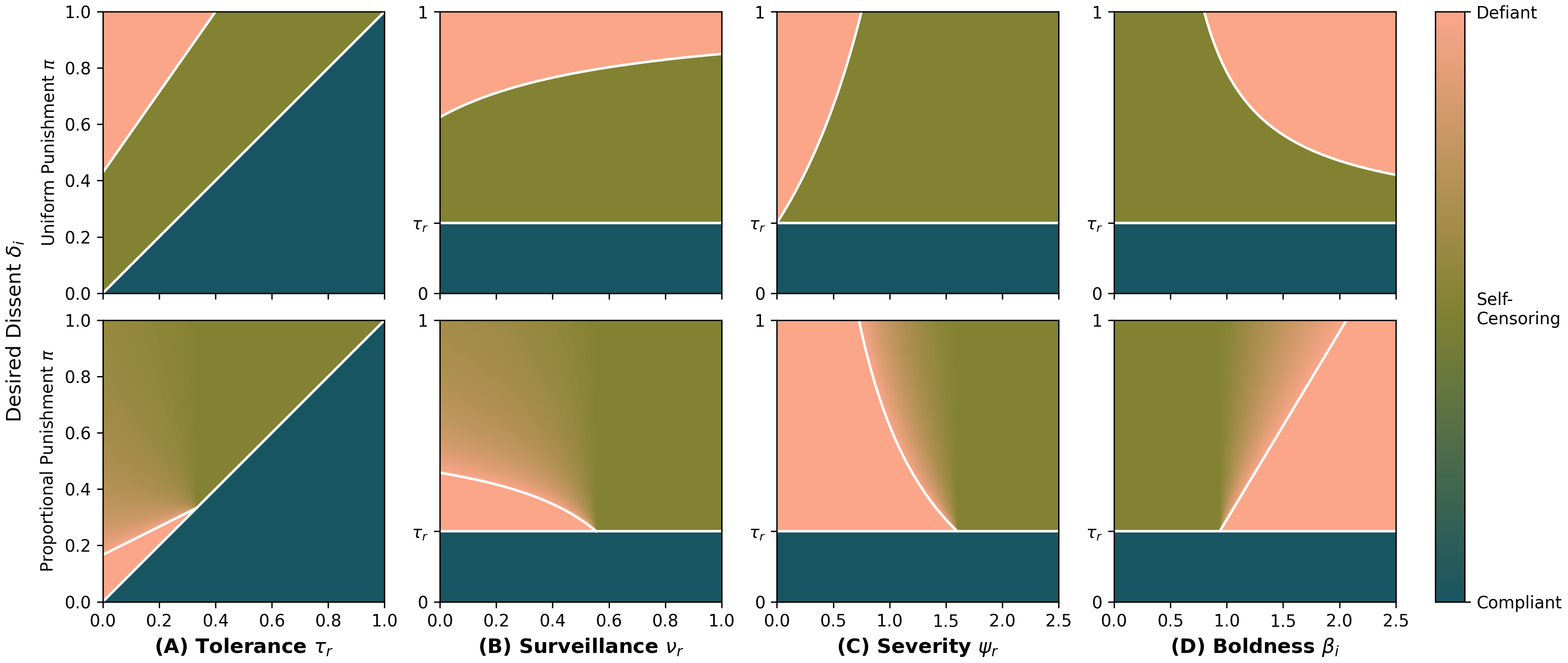}
    \caption{\textit{Phase Diagram of Individuals' Optimal Behaviors by Parameter.}
    Compliant (teal), self-censoring (olive), and defiant (peach) phases of individuals' optimal actions are shown as a function of \textbf{(A)} the authority's tolerance $\tolerate{r}$, \textbf{(B)} the authority's surveillance $\surveil{r}$, \textbf{(C)} the authority's severity $\severity{r}$, and \textbf{(D)} the individual's boldness $\boldness{i}$ vs.\ their desired dissent $\desire{i}$ for both uniform (top row) and proportional (bottom row) punishment.
    The phase boundaries delineating defiance (peach) from self-censorship (olive) occur at critical dissent values $d_{i,r}^{\text{uni}}$ (\eqtext~\ref{eq:optaction:dcon}, top row) and $d_{i,r}^{\text{pro}}$ (\eqtext~\ref{eq:optaction:dpro}, bottom row).
    The specific parameter values used to obtain this visualization were chosen to show all possible phases of behavior in each pairwise sweep, but they are representative of a wider range of parameter values.}
    \label{fig:phases}
\end{figure}

Contemporary examples provide natural illustrations of the phases of behavior predicted by our model (\figtext~\ref{fig:phases}).
As a simple example, after evidence of NSA/PRISM online surveillance was widely publicized in June 2013 (i.e., recognition of increased surveillance without a corresponding change in tolerance or severity), Internet traffic to Wikipedia articles on privacy-sensitive topics significantly and immediately decreased~\cite{Penney2016-chillingeffects}, demonstrating a self-censoring chilling effect consistent with our model's predictions (\figtext~\ref{fig:phases}B).
Survey results from 2017 by the same author found that personal legal threats (i.e., credible threats of high-severity punishment) and guarantees of government surveillance drove greater rates of self-censorship than relevant statues communicating the authority's tolerance~\cite{Penney2017-internetsurveillance}, aligning qualitatively with the roles of these parameters in \eqstext~\ref{eq:optaction:dcon} and~\ref{eq:optaction:dpro} (\figstext~\ref{fig:phases}A--C).

Both China and Russia can be assumed to have low tolerance of dissent and high-severity proportional punishments.
Further, both regimes have made significant investments in their digital infrastructure to achieve self-censorship, as illustrated by the famous analogy of the ``anaconda in the chandelier''~\cite{Link2002-chinaanaconda} and China's 2019 ranking as the most digitally repressive country in the world~\cite{Feldstein2021-risedigital}.
In these settings, our model predicts that only low levels of dissent are commonly expressed, as small infractions are both less likely to be observed and have low expected punishment cost (\figtext~\ref{fig:optaction:prop}A).
For example, the popular 2021 meme of ``Lying Flat'' in China expressed disaffection with hustle culture and contemporary globalization, but was largely confined to online spaces where the stakes were low~\cite{Zhou2023-lyingflat}.
This is reminiscent of the dictator's dilemma~\cite{Wintrobe1998-politicaleconomy} which suggests that authoritarian repression leads to distorted information feedback, creating space for low-risk dissent to flourish.
However, Russia's recently-deployed TSPU~\cite{Epifanova2021-subjugatingrunet,Xue2022-tspurussia} has even deeper surveillance and more granular censorship controls than China's Great Firewall and related censorship technologies~\cite{Ensafi2015-analyzinggreat,Roberts2018-censoreddistraction,Fedasiuk2021-buyingsilence}.
Our model predicts that, holding other parameters constant, Russia's increased surveillance capabilities will drive greater rates of self-censorship than in China (\figtext~\ref{fig:phases}B).

In the social media space where surveillance is high, a study that measured self-censorship on Twitter during the Hong Kong protests of 2022--23~\cite{Wang2023-selfcensorshiplaw} found that after the passage of a new national security law (i.e., a decrease in tolerance and increase in severity), Hong Kong users self-censored by deleting or restricting their accounts, removing old posts, and moving away from politically sensitive topics (\figtext~\ref{fig:phases}A,C).
In the opposite direction (increased tolerance and no severity), users of Gab.com accelerated rates of hate speech after establishing that Gab's content moderation policies tolerated such behavior and did not threaten any serious consequences~\cite{Mathew2020-hatebegets}.
Similarly, Meta's January 2025 reversal of content moderation policies is an example of simultaneously increasing tolerance and removing all severity \cite{Kaplan2025-morefree}.

\subsection{An Adaptive Authority}

The authority's utility minimizes both the population's total expression of dissent (i.e., its \textit{political cost}) and its \textit{punishment cost} (\eqtext~\ref{eq:authutility}).
Perhaps unsurprisingly, there always exists a draconian policy for the authority to suppress all dissent through self-censorship without having to enact any punishments.
Formally, given any population's desired dissent and boldness parameters $\{(\desire{i}, \boldness{i})\}_{i=1}^n$, one utility-maximizing policy for the authority is
\begin{align}
    \tolerate{}^* &= 0, \nonumber \\
    \surveil{}^* &= 1, \text{ and} \label{eq:draconian} \\
    \severity{}^* &> \left\{\begin{array}{cl}
        \max_i\{\desire{i} \cdot \boldness{i}\} & \text{if $\punish$ is uniform;} \\
        \max_i\{\boldness{i}\} & \text{if $\punish$ is proportional.}
    \end{array}\right.\nonumber
\end{align}
In words, this draconian policy tolerates nothing, observes everything, and threatens to punish so extremely that all individuals completely self-censor, expressing no dissent at all.
This policy is, in fact, a Nash equilibrium for the leader--follower game in which the authority acts first and the individuals act second, without collusion and under the assumption that the authority's threat of punishment is fully credible.
The lower bounds on severity $\severity{}^*$ result from \eqstext~\ref{eq:optaction:uniform}--\ref{eq:optaction:dpro} when $\tolerate{r} = 0$ and $\surveil{r} = 1$; see SI Appendix for a derivation.
Notably, the draconian severity $\severity{}^*$ for uniform punishment is at most that of proportional punishment.

In reality, it is unlikely that an authority knows all individuals' desired dissent and boldness parameters a priori or exists in a political context where such an extreme policy could be introduced immediately.
Thus, we also consider how an initially moderate authority might incrementally adapt its policies over time to increase its utility, using feedback through the political and punishment costs incurred in each round to determine how to change its policy.
A round is the time it takes for the authority to observe dissent and enact a new policy, which we take to be roughly days to weeks.

We frame the authority's policy adaptation as the local greedy process of random mutation hill climbing~\cite{Mitchell1993-whenwill}.
The authority is initialized with a moderate policy $(\tolerate{0}, \severity{0}, \surveil{0})$ chosen uniformly at random in $[0, 1]^3$.
Subsequently, in each round, it chooses one of its three parameters at random, mutates that parameter to a new value within an $\varepsilon$-radius of its original value, and adopts that as the new policy.
If the new policy decreases the authority's total cost, it is retained; otherwise, the mutation is reverted (see SI Appendix, Algorithm 1 for details).
This simple rule captures incremental policy change by an initially moderate authority, ignoring any exogenous shocks.

\begin{figure}[t]
    \centering
    \includegraphics[width=0.55\textwidth]{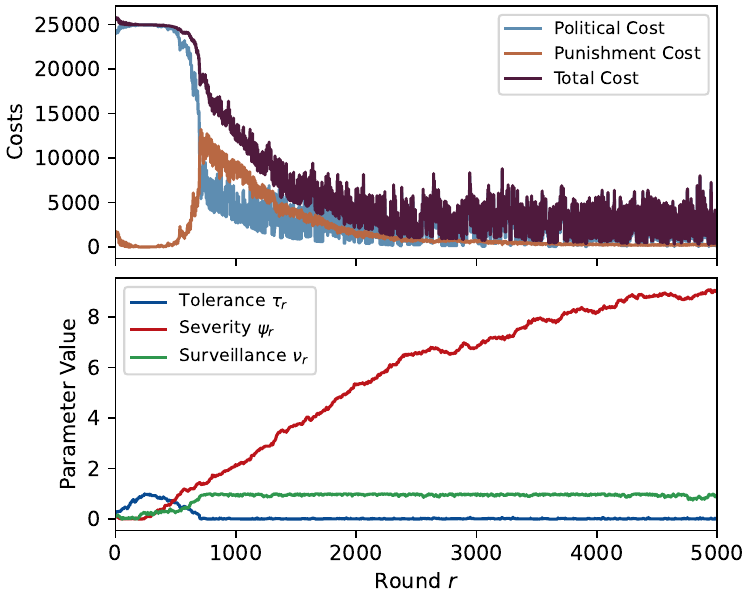}
    \caption{\textit{Time Evolution of an Adaptive Authority.}
    The adaptive authority's costs (top) and parameters (bottom) over \numprint{5000} rounds with uniform punishment and adamancy $\adamancy =$ 1.
    The population has $n =$ \numprint{100000} individuals whose desired dissents are sampled independently from a truncated exponential distribution with mean 0.25 and bounds [0, 1].
    Likewise, their boldness constants are sampled independently from an exponential distributions with mean 2.0.
    Recall that the authority's tolerance $\tolerate{r}$ and surveillance $\surveil{r}$ (bottom, blue and green, resp.) are bounded in $[0, 1]$ while its severity $\severity{r}$ (bottom, red) can be any positive value.
    Corresponding results for proportional punishment can be found in SI Appendix, \figtext~\ref{si:fig:adaptiveauthorityevo}.}
    \label{fig:adaptiveauthorityevo}
\end{figure}

\figtext~\ref{fig:adaptiveauthorityevo} shows an example run of this adaptive authority's evolving policies and corresponding costs over time.
Initially ($r = 0$), the authority has low tolerance (\figtext~\ref{fig:adaptiveauthorityevo}, bottom, blue).
With insufficient severity and surveillance to back that preference (\figtext~\ref{fig:adaptiveauthorityevo}, bottom, red and green), the authority incurs high political cost as the entire population expresses its desired dissent (\figtext~\ref{fig:adaptiveauthorityevo}, top, blue).
Perhaps counterintuitively, the authority does not respond by immediately raising its severity and surveillance, as this would saddle it with the costly task of punishing the vast majority of the population.
Instead, it keeps severity and surveillance low while gradually increasing tolerance until all dissent is allowed (\figtext~\ref{fig:adaptiveauthorityevo}, bottom, blue shows $\tolerate{r} = 1$ near round $r \approx 250$).
The authority then reverses course ($250 \leq r \leq 750$), gradually decreasing its tolerance while increasing its severity and surveillance.
This causes a cascade of self-censorship: High-dissent individuals start by giving up just a small amount of expressed dissent to stay within the authority's tolerance, but then capitulate further and further to match the authority's decreasing tolerance.
As severity and surveillance increase simultaneously, it becomes too costly to regain defiance.
Individuals who are sufficiently bold remain defiant through the cascade, causing an initial spike in the authority's punishment costs (\figtext~\ref{fig:adaptiveauthorityevo}, top, orange, $r \approx 750$), but even they are progressively persuaded to self-censor by continuously increasing punishment severity and high surveillance (\figtext~\ref{fig:adaptiveauthorityevo}, bottom, red and green, $r \geq \numprint{1000}$).
By round $r = \numprint{5000}$, the authority approximates the draconian policy, achieving near-total suppression of dissent.
Similar patterns occur across wide parameter ranges (\figtext~\ref{fig:adaptiveauthoritysweep}) and under both punishment functions (see SI Appendix, \figtext~\ref{si:fig:adaptiveauthorityevo}).

The sequence shown in \figtext~\ref{fig:adaptiveauthorityevo} is reminiscent of Mao's 1956--57 ``Hundred Flowers Campaign'' in which outspoken dissident elites were encouraged to freely criticize the government for the dual purpose of relieving tensions and strengthening the socialist system~\cite{Macfarquhar1989-hundredflowers}.
After a brief period of escalating critique that culminated in popular demands for power sharing~\cite{Lin2006-transformationchinese}---i.e., large political costs resulting from high tolerance---Mao reversed course with the Anti-Rightist Movement, sending thousands of outspoken dissidents to penal camps and labor exile---i.e., a policy of little-to-no tolerance and high severity---dissuading those who remained from expressing further dissent.

\begin{figure}[t]
    \centering
    \includegraphics[width=0.6\textwidth]{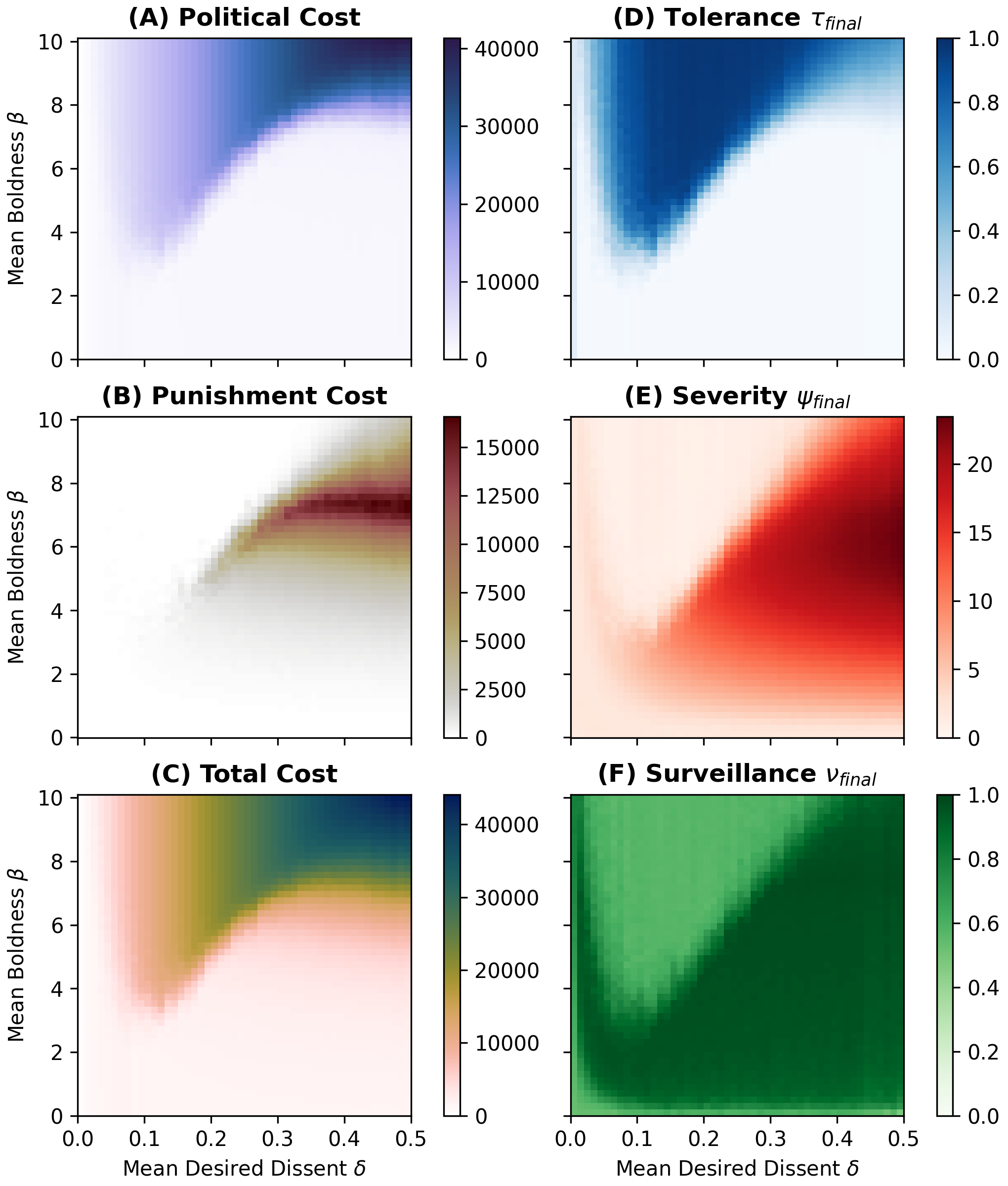}
    \caption{\textit{Authority Costs and Policies After Adaptation Under Uniform Punishment.}
    The authority's average final \textbf{(A)}--\textbf{(C)} costs and \textbf{(D)}--\textbf{(F)} parameter values after \numprint{10000} rounds of random mutation hill climbing adaptation across 50 independent trials per $(\desire{}, \boldness{})$ pair, where $\desire{} \in [0.005, 0.495]$ and $\boldness{} \in [0.1, 10]$ are the means of the population's exponential distributions of desired dissents and boldness constants, respectively.
    Each trial involves $n =$ \numprint{100000} individuals and the authority uses uniform punishment $\punish$ with adamancy $\adamancy =$ 1.
    Corresponding results for proportional punishment can be found in SI Appendix, \figtext~\ref{si:fig:adaptiveauthoritysweep}.}
    \label{fig:adaptiveauthoritysweep}
\end{figure}

Although a draconian policy always exists, the adaptive authority's probability of discovering one depends on the population's desired dissent and boldness distributions (\figtext~\ref{fig:adaptiveauthoritysweep}).
In particular, \figtext~\ref{fig:adaptiveauthoritysweep}C shows a curved boundary delineating a high-boldness region where the authority incurs large total costs from a lower-boldness region of minimal total costs.
In the latter, the authority has discovered a near-draconian policy that neither tolerates dissent (\figtext~\ref{fig:adaptiveauthoritysweep}A,D, lower region) nor enacts much punishment (\figtext~\ref{fig:adaptiveauthoritysweep}B, lower region); instead, with high surveillance and severity that increases with population boldness (\figtext~\ref{fig:adaptiveauthoritysweep}E--F, lower region), it achieves mass self-censorship.
On the other hand, when the population's mean desired dissent is moderate ($\delta \in [0.05, 0.35]$) and mean boldness $\boldness{}$ is sufficiently high, the authority almost never discovers a suppressing policy (\figtext~\ref{fig:adaptiveauthoritysweep}C, upper region).
Instead, the authority tolerates essentially all dissent (\figtext~\ref{fig:adaptiveauthoritysweep}A,D, upper region), reducing severity and making surveillance irrelevant (\figtext~\ref{fig:adaptiveauthoritysweep}E--F, upper region).
This bifurcation of suppressing vs.\ tolerating policies also occurs under proportional punishment (SI Appendix, \figtext~\ref{si:fig:adaptiveauthoritysweep}), with a critical mean boldness $\boldness{}^*$ separating the two.

Between these stable outcomes of suppression through self-censorship and tolerance of all dissent is a small region of realized punishment for moderate-high desired dissent, moderate-boldness populations (\figtext~\ref{fig:adaptiveauthoritysweep}B, maroon region).
This is not a distinct third outcome, but rather a snapshot of transient dynamics in which an authority has successfully driven low-desire/boldness individuals to self-censor but must punish numerous high-desire/boldness individuals who defiantly express dissent (similar to \figtext~\ref{fig:adaptiveauthorityevo}, $r \approx 750$).
Even though the authority will eventually achieve near-total suppression of these populations, the time it takes to do so depends primarily on the boldness of the population.
Thus, while no dissent movement can last forever against increasingly severe authoritarian punishment, increasing a population's mean boldness has a super-linear effect on the length of time that population can sustain dissent and attempt to enact change (SI Appendix, \figtext~\ref{si:fig:suppressiontime}).

\section{Discussion}

Individuals aim to express a desired level of dissent while avoiding punishment.  This leads them to be compliant, self-censoring, or defiant depending on their and the authority's parameters (\figstext~\ref{fig:optaction:uniform}--\ref{fig:optaction:prop}).
Under uniform punishment, when all infractions are punished identically, self-censorship occurs when individuals have an intermediate level of desired dissent,  while a sufficiently high desire leads to defiance (\figtext~\ref{fig:optaction:uniform}A).
Surprisingly, under proportional punishment, the pattern is reversed, with moderate individuals acting as desired while extreme individuals partially self-censor (\figtext~\ref{fig:optaction:prop}A).
In both settings, a critical value of desired dissent separates self-censorship from defiance (\eqstext~\ref{eq:optaction:dcon} and~\ref{eq:optaction:dpro}).
As a fraction of the population, self-censorship increases as the authority's tolerance, severity, or surveillance increases, or as individual boldness decreases (\figtext~\ref{fig:phases}).
There is always a draconian policy with which the authority achieves total self-censorship (\eqtext~\ref{eq:draconian}).
However, when an initially moderate authority adapts incrementally, the population's mean boldness---i.e., individuals' willingness to express dissent in spite of punishment---determines whether the authority suppresses dissent and how long it takes to do so (\figstext~\ref{fig:adaptiveauthorityevo}--\ref{fig:adaptiveauthoritysweep},~\ref{si:fig:adaptiveauthorityevo}--\ref{si:fig:suppressiontime}).

Neither an authority's ability to observe dissent nor self-censorship are new phenomena.
For example, leading a mass gathering or publishing a critical article with a real byline were always visible actions, and many countries' laws prohibit sedition and blasphemy and support substantial self-censorship. 
But modern surveillance technologies have blurred the lines between public and private actions, potentially exposing even minor expressions of dissent to scrutiny and punishment.
As ubiquitous surveillance becomes less an issue of technical capability than one of willingness, norms, and regulations~\cite{Kalluri2025-computervisionresearch,Loewenstein2023-palestinelaboratory,Xue2022-tspurussia,Feldstein2021-risedigital}, it will be increasingly challenging to preserve a culture of free expression (\figstext~\ref{fig:phases}B and~\ref{fig:adaptiveauthorityevo}).

The model shows that the most powerful lever a population has for sustaining dissent is its boldness, i.e., the strength of individuals' preferences to act according to their desired dissent.
Under uniform punishment, high boldness is a necessary condition for any individual to express dissent (\eqtext~\ref{eq:optaction:uniform}); under proportional punishment and perfect surveillance ($\surveil{r} = 1$), the boldness-to-severity ratio completely determines whether defiance or self-censorship occurs (\figstext~\ref{fig:optaction:prop}B--C).
More importantly, the larger a population's mean boldness, the more slowly an initially moderate, locally-adaptive authority discovers a policy that suppresses dissent, if it ever does (\figstext~\ref{fig:adaptiveauthoritysweep},~\ref{si:fig:suppressiontime}).
For example, early and widespread public backlash to privacy-invasive technologies such as China's Green Dam Youth Escort system in 2009~\cite{Yang2011-abortedgreen} and the recent Microsoft Recall feature in Windows 11~\cite{Cunningham2024-microsoftreworking,Cunningham2024-windowsrecall} significantly postponed both projects, even though both eventually reemerged in modified forms.
Threat of punishment may encourage individuals to self-censor, but our results suggest that an adaptive authority does not enact large-scale punishment until after the masses have become compliant (\figtext~\ref{fig:adaptiveauthorityevo}).
It is simply too costly for the authority to punish everyone, but once the vast majority of individuals become compliant or self-censoring, it is cost-effective for the authority to enact increasingly severe punishments against remaining dissidents.
This is reminiscent of the often-quoted first lesson of Snyder's \textit{On Tyranny}: ``Do not obey in advance''~\cite{Snyder2017-tyrannytwenty}.
Preemptive surrender through self-censorship, before punishment is imposed, is a fast path to authoritarian control.
Thus, a population's persistent dissent in spite of threatened punishment early on may deter an authority from adopting more extreme policies.

Promising directions for future work include extending the analyses and simulations to include individuals that adapt their desired level of dissent and boldness values over time.
These parameter adaptations could be strategic in response to the authority's changing policies, or they could arise from opinion dynamics with individuals arranged in networks of social influence akin to Penney's theory of chilling effects as social conformity~\cite{Penney2022-understandingchilling}.
Our preliminary results (see SI Appendix) suggest that assimilative opinion dynamics alone---whether on desired dissent, boldness, or both---are insufficient for a small minority of high-desire, high-boldness dissidents to drive a larger cascade of defiance akin to Chenoweth and Stephan's empirical ``3.5\% Rule''~\cite{Chenoweth2011-whycivil,Chenoweth2021-civilresistance}.
Further investigation could characterize more carefully which behaviors and parameters are most salient in determining when a networked population behaves differently from the setting we considered here.

\begin{acks}
    We thank our reviewers for their valuable feedback and pointers to relevant contemporary examples and empirical studies of self-censorship.
    We also thank Jedidiah Crandall for helpful discussions about Internet censorship technologies, David Peterson for his inputs on the preliminary networked population simulations, Anish Nahar for his preliminary simulation results, and Garrett Parzych for his preliminary analyses of adaptive authorities.
    J.J.D.\ is supported in part by the NSF (CCF-2312537).
    S.F.\ is supported in part by the NSF (CCF-2211750), ARPA-H (SP4701-23-C-0074), and the Santa Fe Institute.
    R.A.\ thanks the University of Michigan for research support.
\end{acks}

\bibliographystyle{ACM-Reference-Format}
\bibliography{ref}


\begin{thebibliography}{57}


\ifx \showCODEN    \undefined \def \showCODEN     #1{\unskip}     \fi
\ifx \showISBNx    \undefined \def \showISBNx     #1{\unskip}     \fi
\ifx \showISBNxiii \undefined \def \showISBNxiii  #1{\unskip}     \fi
\ifx \showISSN     \undefined \def \showISSN      #1{\unskip}     \fi
\ifx \showLCCN     \undefined \def \showLCCN      #1{\unskip}     \fi
\ifx \shownote     \undefined \def \shownote      #1{#1}          \fi
\ifx \showarticletitle \undefined \def \showarticletitle #1{#1}   \fi
\ifx \showURL      \undefined \def \showURL       {\relax}        \fi
\providecommand\bibfield[2]{#2}
\providecommand\bibinfo[2]{#2}
\providecommand\natexlab[1]{#1}
\providecommand\showeprint[2][]{arXiv:#2}

\bibitem[Ackerman and Rodal(2008)]%
        {Ackerman2008-strategicdimensions}
\bibfield{author}{\bibinfo{person}{Peter Ackerman} {and} \bibinfo{person}{Berel
  Rodal}.} \bibinfo{year}{2008}\natexlab{}.
\newblock \showarticletitle{The {{Strategic Dimensions}} of {{Civil
  Resistance}}}.
\newblock \bibinfo{journal}{\emph{Survival}} \bibinfo{volume}{50},
  \bibinfo{number}{3} (\bibinfo{year}{2008}), \bibinfo{pages}{111--126}.
\newblock
\href{https://doi.org/10.1080/00396330802173131}{doi:\nolinkurl{10.1080/00396330802173131}}


\bibitem[{Amnesty International}(2023)]%
        {AmnestyInternational2023-automatedapartheid}
\bibfield{author}{\bibinfo{person}{{Amnesty International}}.}
  \bibinfo{year}{2023}\natexlab{}.
\newblock \bibinfo{booktitle}{\emph{{Israel and Occupied Palestinian
  Territories: Automated Apartheid: How Facial Recognition Fragments,
  Segregates and Controls Palestinians in the OPT}}}.
\newblock \bibinfo{type}{{T}echnical {R}eport} MDE 15/6701/2023.
\newblock


\bibitem[Bar-Tal(2017)]%
        {Bar-Tal2017-selfcensorshipsociopoliticalpsychological}
\bibfield{author}{\bibinfo{person}{Daniel Bar-Tal}.}
  \bibinfo{year}{2017}\natexlab{}.
\newblock \showarticletitle{Self-{{Censorship}} as a
  {{Socio}}-{{Political}}-{{Psychological Phenomenon}}: {{Conception}} and
  {{Research}}}.
\newblock \bibinfo{journal}{\emph{Political Psychology}} \bibinfo{volume}{38},
  \bibinfo{number}{S1} (\bibinfo{year}{2017}), \bibinfo{pages}{37--65}.
\newblock
\href{https://doi.org/10.1111/pops.12391}{doi:\nolinkurl{10.1111/pops.12391}}


\bibitem[Chenoweth(2021)]%
        {Chenoweth2021-civilresistance}
\bibfield{author}{\bibinfo{person}{Erica Chenoweth}.}
  \bibinfo{year}{2021}\natexlab{}.
\newblock \bibinfo{booktitle}{\emph{Civil {{Resistance}}: {{What Everyone
  Needs}} to {{Know}}}}.
\newblock \bibinfo{publisher}{Oxford University Press}, \bibinfo{address}{New
  York, NY, USA}.
\newblock


\bibitem[Chenoweth and Stephan(2011)]%
        {Chenoweth2011-whycivil}
\bibfield{author}{\bibinfo{person}{Erica Chenoweth} {and}
  \bibinfo{person}{Maria~J. Stephan}.} \bibinfo{year}{2011}\natexlab{}.
\newblock \bibinfo{booktitle}{\emph{Why {{Civil Resistance Works}}: {{The
  Strategic Logic}} of {{Nonviolent Conflict}}}}.
\newblock \bibinfo{publisher}{Columbia University Press}, \bibinfo{address}{New
  York, NY, USA}.
\newblock


\bibitem[Chowdhury et~al\mbox{.}(2022)]%
        {Chowdhury2022-examiningfactors}
\bibfield{author}{\bibinfo{person}{Farhan~Asif Chowdhury},
  \bibinfo{person}{Dheeman Saha}, \bibinfo{person}{Md~Rashidul Hasan},
  \bibinfo{person}{Koustuv Saha}, {and} \bibinfo{person}{Abdullah Mueen}.}
  \bibinfo{year}{2022}\natexlab{}.
\newblock \showarticletitle{{Examining Factors Associated with Twitter Account
  Suspension Following the 2020 U.S. Presidential Election}}. In
  \bibinfo{booktitle}{\emph{Proceedings of the 2021 IEEE/ACM International
  Conference on Advances in Social Networks Analysis and Mining}}.
  \bibinfo{publisher}{ACM}, \bibinfo{address}{Virtual Event, Netherlands},
  \bibinfo{pages}{607--612}.
\newblock
\href{https://doi.org/10.1145/3487351.3492715}{doi:\nolinkurl{10.1145/3487351.3492715}}


\bibitem[Common(2020)]%
        {Common2020-fearreaper}
\bibfield{author}{\bibinfo{person}{MacKenzie~F. Common}.}
  \bibinfo{year}{2020}\natexlab{}.
\newblock \showarticletitle{{Fear the Reaper: How Content Moderation Rules Are
  Enforced on Social Media}}.
\newblock \bibinfo{journal}{\emph{International Review of Law, Computers \&
  Technology}} \bibinfo{volume}{34}, \bibinfo{number}{2}
  (\bibinfo{year}{2020}), \bibinfo{pages}{126--152}.
\newblock
\href{https://doi.org/10.1080/13600869.2020.1733762}{doi:\nolinkurl{10.1080/13600869.2020.1733762}}


\bibitem[Como and Fagnani(2016)]%
        {Como2016-localaveraging}
\bibfield{author}{\bibinfo{person}{Giacomo Como} {and} \bibinfo{person}{Fabio
  Fagnani}.} \bibinfo{year}{2016}\natexlab{}.
\newblock \showarticletitle{{From Local Averaging to Emergent Global Behaviors:
  The Fundamental Role of Network Interconnections}}.
\newblock \bibinfo{journal}{\emph{Systems \& Control Letters}}
  \bibinfo{volume}{95} (\bibinfo{year}{2016}), \bibinfo{pages}{70--76}.
\newblock
\href{https://doi.org/10.1016/j.sysconle.2016.02.003}{doi:\nolinkurl{10.1016/j.sysconle.2016.02.003}}


\bibitem[Cunningham(2024a)]%
        {Cunningham2024-microsoftreworking}
\bibfield{author}{\bibinfo{person}{Andrew Cunningham}.}
  \bibinfo{year}{2024}\natexlab{a}.
\newblock \showarticletitle{Microsoft Is Reworking {{Recall}} after Researchers
  Point out Its Security Problems}.
\newblock \bibinfo{journal}{\emph{Ars Technica}} (\bibinfo{year}{2024}).
\newblock


\bibitem[Cunningham(2024b)]%
        {Cunningham2024-windowsrecall}
\bibfield{author}{\bibinfo{person}{Andrew Cunningham}.}
  \bibinfo{year}{2024}\natexlab{b}.
\newblock \showarticletitle{Windows {{Recall}} Demands an Extraordinary Level
  of Trust That {{Microsoft}} Hasn't Earned}.
\newblock \bibinfo{journal}{\emph{Ars Technica}} (\bibinfo{year}{2024}).
\newblock


\bibitem[{Cyberspace Administration of China}(2015)]%
        {CAC2015-guangdonginternet}
\bibfield{author}{\bibinfo{person}{{Cyberspace Administration of China}}.}
  \bibinfo{year}{2015}\natexlab{}.
\newblock \bibinfo{title}{{``This year, the Guangdong Internet Police has dealt
  with 269,000 pieces of illegal information of all kinds''}}.
\newblock
  \bibinfo{howpublished}{\url{https://www.cac.gov.cn/2015-09/28/c\_1116702824.htm}}.
\newblock


\bibitem[Ensafi et~al\mbox{.}(2015)]%
        {Ensafi2015-analyzinggreat}
\bibfield{author}{\bibinfo{person}{Roya Ensafi}, \bibinfo{person}{Philipp
  Winter}, \bibinfo{person}{Abdullah Mueen}, {and} \bibinfo{person}{Jedidiah~R.
  Crandall}.} \bibinfo{year}{2015}\natexlab{}.
\newblock \showarticletitle{{Analyzing the Great Firewall of China Over Space
  and Time}}. In \bibinfo{booktitle}{\emph{Proceedings on Privacy Enhancing
  Technologies}}, Vol.~\bibinfo{volume}{2015}. \bibinfo{pages}{61--76}.
\newblock
\href{https://doi.org/10.1515/popets-2015-0005}{doi:\nolinkurl{10.1515/popets-2015-0005}}


\bibitem[Epifanova(2021)]%
        {Epifanova2021-subjugatingrunet}
\bibfield{author}{\bibinfo{person}{Alena Epifanova}.}
  \bibinfo{year}{2021}\natexlab{}.
\newblock \showarticletitle{{Subjugating RuNet: The Kremlin's New Levers of
  Control over Elections and Society}}.
\newblock \bibinfo{journal}{\emph{{The Russia File}}} (\bibinfo{year}{2021}).
\newblock


\bibitem[Fedasiuk(2021)]%
        {Fedasiuk2021-buyingsilence}
\bibfield{author}{\bibinfo{person}{Ryan Fedasiuk}.}
  \bibinfo{year}{2021}\natexlab{}.
\newblock \showarticletitle{{Buying Silence: The Price of Internet Censorship
  in China}}.
\newblock \bibinfo{journal}{\emph{China Brief}} \bibinfo{volume}{21},
  \bibinfo{number}{1} (\bibinfo{year}{2021}), \bibinfo{pages}{19--25}.
\newblock


\bibitem[Feldstein(2021)]%
        {Feldstein2021-risedigital}
\bibfield{author}{\bibinfo{person}{Steven Feldstein}.}
  \bibinfo{year}{2021}\natexlab{}.
\newblock \bibinfo{booktitle}{\emph{The {{Rise}} of {{Digital Repression}}:
  {{How Technology Is Reshaping Power}}, {{Politics}}, and {{Resistance}}}}.
\newblock \bibinfo{publisher}{Oxford University Press}, \bibinfo{address}{New
  York, NY, USA}.
\newblock


\bibitem[Gillespie(2018)]%
        {Gillespie2018-custodiansinternet}
\bibfield{author}{\bibinfo{person}{Tarleton Gillespie}.}
  \bibinfo{year}{2018}\natexlab{}.
\newblock \bibinfo{booktitle}{\emph{{Custodians of the Internet: Platforms,
  Content Moderation, and the Hidden Decisions That Shape Social Media}}}.
\newblock \bibinfo{publisher}{Yale University Press}, \bibinfo{address}{New
  Haven, CT, USA}.
\newblock


\bibitem[Horton(2011)]%
        {Horton2011-selfcensorship}
\bibfield{author}{\bibinfo{person}{John Horton}.}
  \bibinfo{year}{2011}\natexlab{}.
\newblock \showarticletitle{Self-{{Censorship}}}.
\newblock \bibinfo{journal}{\emph{Res Publica}} \bibinfo{volume}{17},
  \bibinfo{number}{1} (\bibinfo{year}{2011}), \bibinfo{pages}{91--106}.
\newblock
\href{https://doi.org/10.1007/s11158-011-9145-3}{doi:\nolinkurl{10.1007/s11158-011-9145-3}}


\bibitem[Hsiao and Lim(2020)]%
        {Hsiao2020-versobook}
\bibfield{author}{\bibinfo{person}{Andrew Hsiao} {and} \bibinfo{person}{Audrea
  Lim}.} \bibinfo{year}{2020}\natexlab{}.
\newblock \bibinfo{booktitle}{\emph{{The Verso Book of Dissent: Revolutionary
  Words From Three Millennia of Rebellion and Resistance}}}.
\newblock \bibinfo{publisher}{Verso}, \bibinfo{address}{London, United
  Kingdom}.
\newblock


\bibitem[{Human Rights Watch}(2023)]%
        {HumanRightsWatch2023-metabroken}
\bibfield{author}{\bibinfo{person}{{Human Rights Watch}}.}
  \bibinfo{year}{2023}\natexlab{}.
\newblock \bibinfo{booktitle}{\emph{{Meta's Broken Promises: Systemic
  Censorship of Palestine Content on Instagram and Facebook}}}.
\newblock \bibinfo{type}{{T}echnical {R}eport}.
\newblock


\bibitem[Kalluri et~al\mbox{.}(2025)]%
        {Kalluri2025-computervisionresearch}
\bibfield{author}{\bibinfo{person}{Pratyusha~Ria Kalluri},
  \bibinfo{person}{William Agnew}, \bibinfo{person}{Myra Cheng},
  \bibinfo{person}{Kentrell Owens}, \bibinfo{person}{Luca Soldaini}, {and}
  \bibinfo{person}{Abeba Birhane}.} \bibinfo{year}{2025}\natexlab{}.
\newblock \showarticletitle{Computer-Vision Research Powers Surveillance
  Technology}.
\newblock \bibinfo{journal}{\emph{Nature}} \bibinfo{volume}{643},
  \bibinfo{number}{8070} (\bibinfo{year}{2025}), \bibinfo{pages}{73--79}.
\newblock
\href{https://doi.org/10.1038/s41586-025-08972-6}{doi:\nolinkurl{10.1038/s41586-025-08972-6}}


\bibitem[Kaplan(2025)]%
        {Kaplan2025-morefree}
\bibfield{author}{\bibinfo{person}{Joel Kaplan}.}
  \bibinfo{year}{2025}\natexlab{}.
\newblock \bibinfo{title}{More {{Free Speech}} and {{Fewer Mistakes}}}.
\newblock
\newblock
\shownote{\url{https://about.fb.com/news/2025/01/meta-more-speech-fewer-mistakes/}}.


\bibitem[Lin(2006)]%
        {Lin2006-transformationchinese}
\bibfield{author}{\bibinfo{person}{Chun Lin}.} \bibinfo{year}{2006}\natexlab{}.
\newblock \bibinfo{booktitle}{\emph{The Transformation of {{Chinese}}
  Socialism}}.
\newblock \bibinfo{publisher}{Duke University Press}, \bibinfo{address}{Durham
  [N.C.]}.
\newblock


\bibitem[Link(2002)]%
        {Link2002-chinaanaconda}
\bibfield{author}{\bibinfo{person}{Perry Link}.}
  \bibinfo{year}{2002}\natexlab{}.
\newblock \showarticletitle{{China: The Anaconda in the Chandelier}}.
\newblock \bibinfo{journal}{\emph{The New York Review of Books}}
  \bibinfo{volume}{49}, \bibinfo{number}{6} (\bibinfo{year}{2002}).
\newblock


\bibitem[Loewenstein(2023)]%
        {Loewenstein2023-palestinelaboratory}
\bibfield{author}{\bibinfo{person}{Antony Loewenstein}.}
  \bibinfo{year}{2023}\natexlab{}.
\newblock \bibinfo{booktitle}{\emph{{The Palestine Laboratory: How Israel
  Exports the Technology of Occupation Around the World}}}.
\newblock \bibinfo{publisher}{Verso Books}, \bibinfo{address}{London, United
  Kingdom}.
\newblock


\bibitem[Lokot(2021)]%
        {Lokot2021-protestsquare}
\bibfield{author}{\bibinfo{person}{Tetyana Lokot}.}
  \bibinfo{year}{2021}\natexlab{}.
\newblock \bibinfo{booktitle}{\emph{{Beyond the Protest Square: Digital Media
  and Augmented Dissent}}}.
\newblock \bibinfo{publisher}{Rowman \& Littlefield}, \bibinfo{address}{Lanham,
  MD, USA}.
\newblock


\bibitem[Loucaides(2024)]%
        {Loucaides2024-changingface}
\bibfield{author}{\bibinfo{person}{Darren Loucaides}.}
  \bibinfo{year}{2024}\natexlab{}.
\newblock \showarticletitle{{The Changing Face of Protest}}.
\newblock \bibinfo{journal}{\emph{{Rest of World}}} (\bibinfo{year}{2024}).
\newblock


\bibitem[Loury(1994)]%
        {Loury1994-selfcensorshippublic}
\bibfield{author}{\bibinfo{person}{Glenn~C. Loury}.}
  \bibinfo{year}{1994}\natexlab{}.
\newblock \showarticletitle{Self-{{Censorship}} in {{Public Discourse}}: {{A
  Theory}} of ``{{Political Correctness}}'' and {{Related Phenomena}}}.
\newblock \bibinfo{journal}{\emph{Rationality and Society}}
  \bibinfo{volume}{6}, \bibinfo{number}{4} (\bibinfo{year}{1994}),
  \bibinfo{pages}{428--461}.
\newblock
\href{https://doi.org/10.1177/1043463194006004002}{doi:\nolinkurl{10.1177/1043463194006004002}}


\bibitem[Macfarquhar et~al\mbox{.}(1989)]%
        {Macfarquhar1989-hundredflowers}
\bibfield{author}{\bibinfo{person}{Roderick Macfarquhar},
  \bibinfo{person}{Timothy Cheek}, {and} \bibinfo{person}{Eugene Wu}.}
  \bibinfo{year}{1989}\natexlab{}.
\newblock \showarticletitle{The {{Hundred Flowers}}}.
\newblock In \bibinfo{booktitle}{\emph{The {{Secret Speeches}} of {{Chairman
  Mao}}}}. \bibinfo{publisher}{Harvard University Asia Center},
  \bibinfo{pages}{105--111}.
\newblock
\href{https://doi.org/10.1163/9781684171125_007}{doi:\nolinkurl{10.1163/9781684171125_007}}


\bibitem[Mathew et~al\mbox{.}(2020)]%
        {Mathew2020-hatebegets}
\bibfield{author}{\bibinfo{person}{Binny Mathew}, \bibinfo{person}{Anurag
  Illendula}, \bibinfo{person}{Punyajoy Saha}, \bibinfo{person}{Soumya Sarkar},
  \bibinfo{person}{Pawan Goyal}, {and} \bibinfo{person}{Animesh Mukherjee}.}
  \bibinfo{year}{2020}\natexlab{}.
\newblock \showarticletitle{Hate Begets {{Hate}}: {{A Temporal Study}} of
  {{Hate Speech}}}.
\newblock \bibinfo{journal}{\emph{Proceedings of the ACM on Human-Computer
  Interaction}} \bibinfo{volume}{4}, \bibinfo{number}{CSCW2}
  (\bibinfo{year}{2020}), \bibinfo{pages}{1--24}.
\newblock
\href{https://doi.org/10.1145/3415163}{doi:\nolinkurl{10.1145/3415163}}


\bibitem[Mitchell et~al\mbox{.}(1993)]%
        {Mitchell1993-whenwill}
\bibfield{author}{\bibinfo{person}{Melanie Mitchell}, \bibinfo{person}{John~H.
  Holland}, {and} \bibinfo{person}{Stephanie Forrest}.}
  \bibinfo{year}{1993}\natexlab{}.
\newblock \showarticletitle{When {{Will}} a {{Genetic Algorithm Outperform Hill
  Climbing}}?}. In \bibinfo{booktitle}{\emph{Advances in {{Neural Information
  Processing Systems}}}}, Vol.~\bibinfo{volume}{6}. \bibinfo{publisher}{Morgan
  Kaufmann}, \bibinfo{pages}{51--58}.
\newblock


\bibitem[Morris(2024)]%
        {Morris2024-whyfacial}
\bibfield{author}{\bibinfo{person}{Chris Morris}.}
  \bibinfo{year}{2024}\natexlab{}.
\newblock \showarticletitle{{Why Facial Recognition Technology Makes These
  Campus Protests Different from Those in the Past}}.
\newblock \bibinfo{journal}{\emph{{Fast Company}}} (\bibinfo{year}{2024}).
\newblock


\bibitem[Mozur(2019)]%
        {Mozur2019-hongkong}
\bibfield{author}{\bibinfo{person}{Paul Mozur}.}
  \bibinfo{year}{2019}\natexlab{}.
\newblock \showarticletitle{{In Hong Kong Protests, Faces Become Weapons}}.
\newblock \bibinfo{journal}{\emph{{The New York Times}}}
  (\bibinfo{year}{2019}).
\newblock


\bibitem[Myers~West(2018)]%
        {MyersWest2018-censoredsuspended}
\bibfield{author}{\bibinfo{person}{Sarah Myers~West}.}
  \bibinfo{year}{2018}\natexlab{}.
\newblock \showarticletitle{{Censored, Suspended, Shadowbanned: User
  Interpretations of Content Moderation on Social Media Platforms}}.
\newblock \bibinfo{journal}{\emph{New Media \& Society}} \bibinfo{volume}{20},
  \bibinfo{number}{11} (\bibinfo{year}{2018}), \bibinfo{pages}{4366--4383}.
\newblock
\href{https://doi.org/10.1177/1461444818773059}{doi:\nolinkurl{10.1177/1461444818773059}}


\bibitem[Nicholas(2022)]%
        {Nicholas2022-sheddinglight}
\bibfield{author}{\bibinfo{person}{Gabriel Nicholas}.}
  \bibinfo{year}{2022}\natexlab{}.
\newblock \bibinfo{booktitle}{\emph{{Shedding Light on Shadowbanning}}}.
\newblock \bibinfo{type}{{T}echnical {R}eport}. \bibinfo{institution}{Center
  for Democracy \& Technology}.
\newblock


\bibitem[{Office of the Director of National Intelligence}(2024)]%
        {ODNI2024-annualstatistical}
\bibfield{author}{\bibinfo{person}{{Office of the Director of National
  Intelligence}}.} \bibinfo{year}{2024}\natexlab{}.
\newblock \bibinfo{booktitle}{\emph{{Annual Statistical Transparency Report
  Regarding the Intelligence Community's Use of National Security Surveillance
  Authorities: Calendar Year 2023}}}.
\newblock \bibinfo{type}{{T}echnical {R}eport}.
\newblock


\bibitem[Penney(2016)]%
        {Penney2016-chillingeffects}
\bibfield{author}{\bibinfo{person}{Jonathon~W. Penney}.}
  \bibinfo{year}{2016}\natexlab{}.
\newblock \showarticletitle{Chilling {{Effects}}: {{Online Surveillance}} and
  {{Wikipedia Use}}}.
\newblock \bibinfo{journal}{\emph{Berkeley Technology Law Journal}}
  \bibinfo{volume}{31}, \bibinfo{number}{1} (\bibinfo{year}{2016}),
  \bibinfo{pages}{118--182}.
\newblock
\href{https://doi.org/10.15779/Z38SS13}{doi:\nolinkurl{10.15779/Z38SS13}}


\bibitem[Penney(2017)]%
        {Penney2017-internetsurveillance}
\bibfield{author}{\bibinfo{person}{Jonathon~W. Penney}.}
  \bibinfo{year}{2017}\natexlab{}.
\newblock \showarticletitle{Internet {{Surveillance}}, {{Regulation}}, and
  {{Chilling Effects Online}}: {{A Comparative Case Study}}}.
\newblock \bibinfo{journal}{\emph{Internet Policy Review}} \bibinfo{volume}{6},
  \bibinfo{number}{2} (\bibinfo{year}{2017}), \bibinfo{pages}{1--24}.
\newblock
\href{https://doi.org/10.14763/2017.2.692}{doi:\nolinkurl{10.14763/2017.2.692}}


\bibitem[Penney(2022)]%
        {Penney2022-understandingchilling}
\bibfield{author}{\bibinfo{person}{Jonathon~W. Penney}.}
  \bibinfo{year}{2022}\natexlab{}.
\newblock \showarticletitle{Understanding {{Chilling Effects}}}.
\newblock \bibinfo{journal}{\emph{Minnesota Law Review}} \bibinfo{volume}{106},
  \bibinfo{number}{3} (\bibinfo{year}{2022}), \bibinfo{pages}{1451--1530}.
\newblock


\bibitem[Pierri et~al\mbox{.}(2023)]%
        {Pierri2023-howdoes}
\bibfield{author}{\bibinfo{person}{Francesco Pierri}, \bibinfo{person}{Luca
  Luceri}, \bibinfo{person}{Emily Chen}, {and} \bibinfo{person}{Emilio
  Ferrara}.} \bibinfo{year}{2023}\natexlab{}.
\newblock \showarticletitle{{How Does Twitter Account Moderation Work? Dynamics
  of Account Creation and Suspension on Twitter during Major Geopolitical
  Events}}.
\newblock \bibinfo{journal}{\emph{EPJ Data Science}} \bibinfo{volume}{12},
  \bibinfo{number}{1} (\bibinfo{year}{2023}), \bibinfo{pages}{43}.
\newblock
\href{https://doi.org/10.1140/epjds/s13688-023-00420-7}{doi:\nolinkurl{10.1140/epjds/s13688-023-00420-7}}


\bibitem[Roberts(2018)]%
        {Roberts2018-censoreddistraction}
\bibfield{author}{\bibinfo{person}{Margaret~E. Roberts}.}
  \bibinfo{year}{2018}\natexlab{}.
\newblock \bibinfo{booktitle}{\emph{Censored: {{Distraction}} and {{Diversion
  Inside China}}'s {{Great Firewall}}}}.
\newblock \bibinfo{publisher}{Princeton University Press}.
\newblock
\href{https://doi.org/10.2307/j.ctvc77b21}{doi:\nolinkurl{10.2307/j.ctvc77b21}}


\bibitem[Schelling(1971)]%
        {Schelling1971-questionscivilian}
\bibfield{author}{\bibinfo{person}{Thomas~C. Schelling}.}
  \bibinfo{year}{1971}\natexlab{}.
\newblock \showarticletitle{Some {{Questions}} on {{Civilian Defense}}}.
\newblock In \bibinfo{booktitle}{\emph{Conflict: {{Violence}} and
  {{Nonviolence}}} (\bibinfo{edition}{1st} ed.)}. \bibinfo{publisher}{{Taylor
  and Francis}}, \bibinfo{address}{London, UK}.
\newblock


\bibitem[Shahid and Vashistha(2023)]%
        {Shahid2023-decolonizingcontent}
\bibfield{author}{\bibinfo{person}{Farhana Shahid} {and}
  \bibinfo{person}{Aditya Vashistha}.} \bibinfo{year}{2023}\natexlab{}.
\newblock \showarticletitle{{Decolonizing Content Moderation: Does Uniform
  Global Community Standard Resemble Utopian Equality or Western Power
  Hegemony?}}. In \bibinfo{booktitle}{\emph{Proceedings of the 2023 CHI
  Conference on Human Factors in Computing Systems}}. \bibinfo{publisher}{ACM},
  \bibinfo{address}{Hamburg, Germany}, \bibinfo{pages}{1--18}.
\newblock
\href{https://doi.org/10.1145/3544548.3581538}{doi:\nolinkurl{10.1145/3544548.3581538}}


\bibitem[Sharp(1973)]%
        {Sharp1973-politicsnonviolent}
\bibfield{author}{\bibinfo{person}{Gene Sharp}.}
  \bibinfo{year}{1973}\natexlab{}.
\newblock \bibinfo{booktitle}{\emph{The {{Politics}} of {{Nonviolent
  Action}}}}.
\newblock \bibinfo{publisher}{Porter Sargent}, \bibinfo{address}{Boston, MA,
  USA}.
\newblock


\bibitem[Shen and Truex(2021)]%
        {Shen2021-searchselfcensorship}
\bibfield{author}{\bibinfo{person}{Xiaoxiao Shen} {and} \bibinfo{person}{Rory
  Truex}.} \bibinfo{year}{2021}\natexlab{}.
\newblock \showarticletitle{In {{Search}} of {{Self-Censorship}}}.
\newblock \bibinfo{journal}{\emph{British Journal of Political Science}}
  \bibinfo{volume}{51}, \bibinfo{number}{4} (\bibinfo{year}{2021}),
  \bibinfo{pages}{1672--1684}.
\newblock
\href{https://doi.org/10.1017/S0007123419000735}{doi:\nolinkurl{10.1017/S0007123419000735}}


\bibitem[Snyder(2017)]%
        {Snyder2017-tyrannytwenty}
\bibfield{author}{\bibinfo{person}{Timothy Snyder}.}
  \bibinfo{year}{2017}\natexlab{}.
\newblock \bibinfo{booktitle}{\emph{On {{Tyranny}}: {{Twenty Lessons}} from the
  {{Twentieth Century}}}}.
\newblock \bibinfo{publisher}{Crown}, \bibinfo{address}{New York, NY, USA}.
\newblock


\bibitem[Toraman et~al\mbox{.}(2022)]%
        {Toraman2022-blacklivesmatter2020}
\bibfield{author}{\bibinfo{person}{Cagri Toraman}, \bibinfo{person}{Furkan {\c
  S}ahinu{\c c}}, {and} \bibinfo{person}{Eyup~Halit Yilmaz}.}
  \bibinfo{year}{2022}\natexlab{}.
\newblock \showarticletitle{{BlackLivesMatter 2020: An Analysis of Deleted and
  Suspended Users in Twitter}}. In \bibinfo{booktitle}{\emph{14th ACM Web
  Science Conference 2022}}. \bibinfo{publisher}{ACM},
  \bibinfo{address}{Barcelona, Spain}, \bibinfo{pages}{290--295}.
\newblock
\href{https://doi.org/10.1145/3501247.3531539}{doi:\nolinkurl{10.1145/3501247.3531539}}


\bibitem[Trump(2025a)]%
        {Trump2025-endingillegal}
\bibfield{author}{\bibinfo{person}{Donald~J. Trump}.}
  \bibinfo{year}{2025}\natexlab{a}.
\newblock \bibinfo{booktitle}{\emph{Ending {{Illegal Discrimination And
  Restoring Merit-Based Opportunity}}}}.
\newblock \bibinfo{type}{Presidential {{Action}}}. \bibinfo{institution}{The
  White House}, \bibinfo{address}{Washington DC, USA}.
\newblock


\bibitem[Trump(2025b)]%
        {Trump2025-endingradical}
\bibfield{author}{\bibinfo{person}{Donald~J. Trump}.}
  \bibinfo{year}{2025}\natexlab{b}.
\newblock \bibinfo{booktitle}{\emph{Ending {{Radical And Wasteful Government
  DEI Programs And Preferencing}}}}.
\newblock \bibinfo{type}{Executive {{Order}}}. \bibinfo{institution}{The White
  House}, \bibinfo{address}{Washington DC, USA}.
\newblock


\bibitem[Trump(2025c)]%
        {Trump2025-restoringfreedom}
\bibfield{author}{\bibinfo{person}{Donald~J. Trump}.}
  \bibinfo{year}{2025}\natexlab{c}.
\newblock \bibinfo{booktitle}{\emph{Restoring {{Freedom Of Speech And Ending
  Federal Censorship}}}}.
\newblock \bibinfo{type}{Executive {{Order}}}. \bibinfo{institution}{The White
  House}, \bibinfo{address}{Washington DC, USA}.
\newblock


\bibitem[Vincent(2020)]%
        {Vincent2020-nypdused}
\bibfield{author}{\bibinfo{person}{James Vincent}.}
  \bibinfo{year}{2020}\natexlab{}.
\newblock \showarticletitle{{NYPD Used Facial Recognition to Track Down Black
  Lives Matter Activist}}.
\newblock \bibinfo{journal}{\emph{{The Verge}}} (\bibinfo{year}{2020}).
\newblock


\bibitem[Wang and Mayer(2023)]%
        {Wang2023-selfcensorshiplaw}
\bibfield{author}{\bibinfo{person}{Mona Wang} {and} \bibinfo{person}{Jonathan
  Mayer}.} \bibinfo{year}{2023}\natexlab{}.
\newblock \showarticletitle{Self-{{Censorship Under Law}}: {{A Case Study}} of
  the {{Hong Kong National Security Law}}}.
\newblock \bibinfo{journal}{\emph{Free and Open Communications on the
  Internet}}  \bibinfo{volume}{1} (\bibinfo{year}{2023}),
  \bibinfo{pages}{46--55}.
\newblock


\bibitem[Wilson and Land(2021)]%
        {Wilson2021-hatespeech}
\bibfield{author}{\bibinfo{person}{Richard~Ashby Wilson} {and}
  \bibinfo{person}{Molly~K. Land}.} \bibinfo{year}{2021}\natexlab{}.
\newblock \showarticletitle{{Hate Speech on Social Media: Content Moderation in
  Context}}.
\newblock \bibinfo{journal}{\emph{Connecticut Law Review}}
  \bibinfo{volume}{52}, \bibinfo{number}{3} (\bibinfo{year}{2021}),
  \bibinfo{pages}{1029--1076}.
\newblock


\bibitem[Wintrobe(1998)]%
        {Wintrobe1998-politicaleconomy}
\bibfield{author}{\bibinfo{person}{Ronald Wintrobe}.}
  \bibinfo{year}{1998}\natexlab{}.
\newblock \bibinfo{booktitle}{\emph{The Political Economy of Dictatorship}}.
\newblock \bibinfo{publisher}{Cambridge University Press},
  \bibinfo{address}{Cambridge, MA, USA}.
\newblock


\bibitem[Xue et~al\mbox{.}(2022)]%
        {Xue2022-tspurussia}
\bibfield{author}{\bibinfo{person}{Diwen Xue}, \bibinfo{person}{Benjamin
  {Mixon-Baca}}, \bibinfo{person}{{ValdikSS}}, \bibinfo{person}{Anna Ablove},
  \bibinfo{person}{Beau Kujath}, \bibinfo{person}{Jedidiah~R. Crandall}, {and}
  \bibinfo{person}{Roya Ensafi}.} \bibinfo{year}{2022}\natexlab{}.
\newblock \showarticletitle{{TSPU: Russia's Decentralized Censorship System}}.
  In \bibinfo{booktitle}{\emph{Proceedings of the 22nd ACM Internet Measurement
  Conference}}. \bibinfo{publisher}{ACM}, \bibinfo{address}{Nice, France},
  \bibinfo{pages}{179--194}.
\newblock
\href{https://doi.org/10.1145/3517745.3561461}{doi:\nolinkurl{10.1145/3517745.3561461}}


\bibitem[Yang(2011)]%
        {Yang2011-abortedgreen}
\bibfield{author}{\bibinfo{person}{Kenneth~C.C. Yang}.}
  \bibinfo{year}{2011}\natexlab{}.
\newblock \showarticletitle{The Aborted {{Green}} Dam-Youth Escort Censor-Ware
  Project in {{China}}: {{A}} Case Study of Emerging Civic Participation in
  {{China}}'s Internet Policy-Making Process}.
\newblock \bibinfo{journal}{\emph{Telematics and Informatics}}
  \bibinfo{volume}{28}, \bibinfo{number}{2} (\bibinfo{year}{2011}),
  \bibinfo{pages}{101--111}.
\newblock
\href{https://doi.org/10.1016/j.tele.2010.07.001}{doi:\nolinkurl{10.1016/j.tele.2010.07.001}}


\bibitem[YouTube(2025)]%
        {YouTube2025-youtubecommunity}
\bibfield{author}{\bibinfo{person}{YouTube}.} \bibinfo{year}{2025}\natexlab{}.
\newblock \bibinfo{title}{{{YouTube Community Guidelines}}: {{Taking Action}}
  on {{Violations}}}.
\newblock
  \bibinfo{howpublished}{https://www.youtube.com/howyoutubeworks/policies/community-guidelines/\#taking-action-on-violations}.
\newblock


\bibitem[Zhou(2023)]%
        {Zhou2023-lyingflat}
\bibfield{author}{\bibinfo{person}{Yanqiu~Rachel Zhou}.}
  \bibinfo{year}{2023}\natexlab{}.
\newblock \showarticletitle{The Lying Flat Movement, Global Youth, and
  Globality: A Case of Collective Reading on {{Reddit}}}.
\newblock \bibinfo{journal}{\emph{Globalizations}} \bibinfo{volume}{20},
  \bibinfo{number}{4} (\bibinfo{year}{2023}), \bibinfo{pages}{679--695}.
\newblock
\href{https://doi.org/10.1080/14747731.2023.2165377}{doi:\nolinkurl{10.1080/14747731.2023.2165377}}


\end{thebibliography}

\appendix

\renewcommand{\thefigure}{S\arabic{figure}}
\setcounter{figure}{0}
\renewcommand{\theequation}{S\arabic{equation}}
\setcounter{equation}{0}

\section{Supporting Information (SI Appendix)}

This Supporting Information details all derivations of analytical results, the experimental setup and reproducibility information for simulation results, and supplementary results not included in the Main Text.
Mathematica source code verifying the derivations and Python source code for the simulations are openly available at \url{https://doi.org/10.5281/zenodo.15150163}.

\subsection{Derivations of Individuals' Optimal Actions}

The following derivations support the analytical results for the optimal action of an individual $I_i$ with desired dissent $\desire{i}$ and boldness $\boldness{i}$.
We assume $I_i$ knows the authority's punishment function $\punish$ but only has (potentially noisy) estimates $\estimatetol{r}$, $\estimatesev{r}$, and $\estimatesur{r}$ of the authority's tolerance, severity, and surveillance in round $r$, respectively.
Formally, we solve for
\begin{equation}
    \action{i}{r}^*
    = \mathrm{argmax}_{\action{i}{r}} \{\E{U_{i,r}}\}
    = \mathrm{argmax}_{\action{i}{r}} \left\{\boldness{i} \cdot (1 - \desire{i} + \action{i}{r}) - \E{\punish\left(\action{i}{r}, \estimatetol{r}, \estimatesev{r}, \estimatesur{r}\right)}\right\},
\end{equation}
where expectation is taken over the authority's probability of observing the action $\action{i}{r}$.

\begin{theorem}[\eqstext~\ref{eq:optaction:uniform}--\ref{eq:optaction:dcon}] \label{thm:optaction:uniform}
    When $\punish$ is \textbf{uniform}, the optimal action for an individual $I_i$ is
    \[\action{i}{r}^* = \left\{\begin{array}{ll}
        \desire{i} & \text{if $\desire{i} \leq \tolerate{r}$} \text{ or $\big(\boldness{i} > (1 - \surveil{r}) \cdot \severity{r}$ and $\desire{i} > d_{i,r}^{\text{uni}}\big)$}; \\
        \tolerate{r} & \text{otherwise},
    \end{array}\right.
    \quad \text{where} \quad
    d_{i,r}^{\text{uni}} = \frac{\boldness{i} \cdot \tolerate{r} + \surveil{r} \cdot \severity{r}}{\boldness{i} - (1 - \surveil{r}) \cdot \severity{r}}.\]
\end{theorem}
\begin{proof}
    When $\punish$ is uniform, we have $\punish\left(\action{i}{r}, \estimatetol{r}, \estimatesev{r}, \estimatesur{r}\right) = \estimatesev{r}$ when $\action{i}{r} > \estimatetol{r}$ and $0$ otherwise, where $\estimatetol{r} \in [0, 1]$ and $\estimatesev{r} > 0$ are the individual's estimates of the authority's tolerance and severity in round $r$, respectively.
    Recall that an individual does not expect to be punished if it acts below its estimation of the authority's tolerance, and otherwise expect to be punished only if it is observed by the authority.
    So the expected punishment is
    \begin{equation}
        \E{\punish\left(\action{i}{r}, \estimatetol{r}, \estimatesev{r}, \estimatesur{r}\right)} = \left\{\begin{array}{ll}
            0 & \text{if $\action{i}{r} \leq \estimatetol{r}$}; \\
            \estimatesev{r} \cdot (\estimatesur{r} + (1 - \estimatesur{r}) \cdot \action{i}{r}) & \text{otherwise.}
        \end{array}\right.
    \end{equation}
    The corresponding expected utility $\E{U_{i,r}}$ is
    \begin{align}
        \E{U_{i,r}} &= \boldness{i} \cdot (1 - \desire{i} + \action{i}{r}) - \E{\punish\left(\action{i}{r}, \estimatetol{r}, \estimatesev{r}, \estimatesur{r}\right)} \\
        &= \left\{\begin{array}{ll}
            \boldness{i} \cdot (1 - \desire{i} + \action{i}{r}) & \text{if $\action{i}{r} \leq \estimatetol{r}$}; \\
            \boldness{i} \cdot (1 - \desire{i} + \action{i}{r}) - \estimatesev{r} \cdot (\estimatesur{r} + (1 - \estimatesur{r}) \cdot \action{i}{r}) & \text{otherwise.}
        \end{array}\right.
    \end{align}
    
    First suppose that $\desire{i} \leq \estimatetol{r}$, i.e., the individual's desired dissent is below its estimation of the authority's tolerance.
    Then $\E{U_{i,r}} = \boldness{i} \cdot (1 - \desire{i} + \action{i}{r})$, which increases linearly with $\action{i}{r}$, so the optimal action in this case is $\action{i}{r}^* = \desire{i}$.
    Intuitively, there is no punishment for $I_i$ acting as it wants to, so it does so.
    
    Now suppose $\desire{i} > \estimatetol{r}$.
    By the same logic as in the previous case, $\E{U_{i,r}}$ increases linearly as $\action{i}{r}$ increases from 0 to $\estimatetol{r}$.
    At $\action{i}{r} = \estimatetol{r}$, $\E{U_{i,r}}$ is discontinuous as the authority begins punishing any actions it observes.
    Then, as $\action{i}{r}$ increases from $\estimatetol{r}$ to $\desire{i}$, $\E{U_{i,r}}$ is both increasing linearly as the individual acts closer to its desire and decreasing linearly as the probability of observation (and subsequent uniform punishment) grows.
    The combination of these effects is linear in $\action{i}{r}$, so the two candidates for optimal action in this case are $\action{i}{r}^* \in \{\estimatetol{r}, \desire{i}\}$; i.e., either the individual should act as dissenting as it can without being punished, or it should accept punishment and act according to its desired dissent.
    We find that
    \begin{equation}
        \E{U_{i,r} \mid \action{i}{r} = \desire{i}} > \E{U_{i,r} \mid \action{i}{r} = \estimatetol{r}}
        \iff \boldness{i} > (1 - \estimatesur{r}) \cdot \estimatesev{r} \text{ and } \desire{i} > \frac{\boldness{i} \cdot \estimatetol{r} + \estimatesur{r} \cdot \estimatesev{r}}{\boldness{i} - (1 - \estimatesur{r}) \cdot \estimatesev{r}}.
    \end{equation}
    
    Combining these two cases and considering the special case when the individuals' estimates are correct yields the theorem.
    If the estimates are noisy or otherwise inaccurate, replacing $\tolerate{r}$, $\severity{r}$, and $\surveil{r}$ with $\estimatetol{r}$, $\estimatesev{r}$, and $\estimatesur{r}$ in the theorem yields the individual's utility-maximizing action with respect to its estimates.
\end{proof}

\begin{theorem}[\eqstext~\ref{eq:optaction:prop}--\ref{eq:optaction:dpro}] \label{thm:optaction:proportional}
    When $\punish$ is \textbf{proportional}, the optimal action for an individual $I_i$ is
    \[\action{i}{r}^* = \left\{\begin{array}{ll}
        \desire{i} & \text{if $\desire{i} \leq \tolerate{r}$ or $(\surveil{r} = 1$ and $\boldness{i} \geq \severity{r})$ or $(\surveil{r} < 1$ and $\desire{i} \leq d_{i,r}^{\text{pro}})$}; \\
        \tolerate{r} & \text{if $\desire{i} > \tolerate{r}$ and $\big((\surveil{r} = 1$ and $\boldness{i} < \severity{r})$ or $(\surveil{r} < 1$ and $\tolerate{r} \geq d_{i,r}^{\text{pro}})\big)$}; \\
        d_{i,r}^{\text{pro}} & \text{otherwise},
    \end{array}\right.\]
    where
    \[d_{i,r}^{\text{pro}} = \frac{1}{2}\left(\tolerate{r} + \frac{\boldness{i} - \surveil{r} \cdot \severity{r}}{(1 - \surveil{r}) \cdot \severity{r}}\right).\]
\end{theorem}
\begin{proof}
    When $\punish$ is proportional (i.e., linear), we have $\punish\left(\action{i}{r}, \estimatetol{r}, \estimatesev{r}, \estimatesur{r}\right) = \estimatesev{r} \cdot (\action{i}{r} - \estimatetol{r})$ when $\action{i}{r} > \estimatetol{r}$ and $0$ otherwise, where again $\estimatetol{r} \in [0, 1]$ and $\estimatesev{r} > 0$ are the individual's estimates of the authority's tolerance and severity in round $r$, respectively.
    In this case, the expected punishment is
    \begin{equation}
        \E{\punish\left(\action{i}{r}, \estimatetol{r}, \estimatesev{r}, \estimatesur{r}\right)} = \left\{\begin{array}{ll}
            0 & \text{if $\action{i}{r} \leq \estimatetol{r}$}; \\
            \estimatesev{r} \cdot (\action{i}{r} - \estimatetol{r}) \cdot (\estimatesur{r} + (1 - \estimatesur{r}) \cdot \action{i}{r}) & \text{otherwise,}
        \end{array}\right.
    \end{equation}
    and the corresponding expected utility $\E{U_{i,r}}$ is
    \begin{align}
        \E{U_{i,r}} &= \boldness{i} \cdot (1 - \desire{i} + \action{i}{r}) - \E{\punish\left(\action{i}{r}, \estimatetol{r}, \estimatesev{r}, \estimatesur{r}\right)} \\
        &= \left\{\begin{array}{ll}
            \boldness{i} \cdot (1 - \desire{i} + \action{i}{r}) & \text{if $\action{i}{r} \leq \estimatetol{i}{r}$}; \\
            \boldness{i} \cdot (1 - \desire{i} + \action{i}{r}) - \estimatesev{r} \cdot (\action{i}{r} - \estimatetol{r}) \cdot (\estimatesur{r} + (1 - \estimatesur{r}) \cdot \action{i}{r}) & \text{otherwise.}
        \end{array}\right.
    \end{align}
    
    As in the case where $\punish$ is uniform, if $\desire{i} \leq \estimatetol{r}$ then the individual believes it can act according to its desired dissent without punishment, so it does so with $\action{i}{r}^* = \desire{i}$.
    
    Suppose instead that $\desire{i} > \estimatetol{r}$.
    Again, $\E{U_{i,r}}$ is linearly increasing as $\action{i}{r}$ increases from $0$ to $\estimatetol{r}$, the interval without punishment.
    As $\action{i}{r}$ increases from $\estimatetol{r}$ to $\desire{i}$, there are two cases.
    If the authority's surveillance is perfect (i.e., $\estimatesur{r} = 1$), then $\E{U_{i,r}}$ increases linearly by acting closer to its dissent desire but decreases linearly as the cost of punishment increases.
    The combination of these effects is linear in $\action{i}{r}$, so the critical values are $\action{i}{r}^* \in \{\estimatetol{r}, \desire{i}\}$.
    In this case, we find that
    \begin{equation}
        \E{U_{i,r} \mid \action{i}{r} = \desire{i}} \geq \E{U_{i,r} \mid \action{i}{r} = \estimatetol{r}} \iff \boldness{i} \geq \estimatesev{r}.
    \end{equation}
    
    Otherwise, if $\estimatesur{r} < 1$, then the individual believes there is some probability that the authority will not observe its action.
    Both the probability of being observed and the subsequent cost of punishment increase linearly with $\action{i}{r}$, meaning that $\E{U_{i,r}}$ increases linearly as it acts closer to its desire but decreases quadratically due to expected punishment.
    In this case, $\frac{\partial^2}{\partial \action{i}{r}} \E{U_{i,r}} < 0$; i.e., expected utility is a concave-down parabola.
    The parabola reaches its maximum value at
    \begin{equation}
        d_{i,r}^{\text{pro}} = \frac{1}{2}\left(\estimatetol{r} + \frac{\boldness{i} - \estimatesur{r} \cdot \estimatesev{r}}{(1 - \estimatesur{r}) \cdot \estimatesev{r}}\right).
    \end{equation}
    The parabola's concavity immediately implies that the optimal action dissent in this case is
    \begin{equation}
        \action{i}{r}^* = \left\{\begin{array}{ll}
            \estimatetol{r} & \text{if $d_{i,r}^{\text{pro}} \leq \estimatetol{r}$}; \\
            \desire{i} & \text{if $d_{i,r}^{\text{pro}} \geq \desire{i}$}; \\
            d_{i,r}^{\text{pro}} & \text{otherwise}.
        \end{array}\right.
    \end{equation}
    
    Combining all of these cases and considering the special case when the individuals' estimates are correct yields the theorem.
    Again, if the estimates are noisy or otherwise inaccurate, replacing $\tolerate{r}$, $\severity{r}$, and $\surveil{r}$ with $\estimatetol{r}$, $\estimatesev{r}$, and $\estimatesur{r}$ in the theorem yields the individual's utility-maximizing action with respect to its estimates.
\end{proof}

\subsection{The Authority's Draconian Policy}

The following derivation yields the authority's utility-maximizing draconian policy that represses all dissent through self-censorship without having to enact any punishments.

\begin{theorem}[\eqtext~\ref{eq:draconian}] \label{thm:draconian}
    Given a population's desired dissent and boldness parameters $\{(\desire{i}, \boldness{i})\}_{i=1}^n$, an optimal policy for the authority is $\tolerate{}^* = 0$, $\surveil{}^* = 1$, and $\severity{}^* \geq \max_i\{\desire{i} \cdot \boldness{i}\}$ if $\punish$ is uniform or $\severity{}^* > \max_i\{\boldness{i}\}$ if $\punish$ is proportional.
\end{theorem}
\begin{proof}
    \eqtext~\ref{eq:authutility} defines the authority's utility as
    \begin{equation}
        U_{A,r} = -\adamancy \cdot \sum_{i=1}^n \action{i}{r} - \sum_{i=1}^n \punish(\action{i}{r}, \tolerate{r}, \severity{r}, \surveil{r}),
    \end{equation}
    where $\adamancy > 0$ is the authority's adamancy balancing the relative importance of expressed dissent compared to the cost of punishing.
    Since individuals' actions $\action{i}{r}$ and the authority's punishments $\punish(\action{i}{r}, \tolerate{r}, \severity{r}, \surveil{r})$ are both nonnegative, the maximum value of $U_{A,r}$ is zero.
    Thus, any policy that achieves $\E{U_{A,r}} = 0$ (where expectation is taken over the authority's probabilities of observing individuals' actions) is optimal.

    If the authority sets its tolerance $\tolerate{r} > 0$, Theorems~\ref{thm:optaction:uniform} and~\ref{thm:optaction:proportional} show that individuals with $\desire{i} \leq \tolerate{r}$ will simply enact their (potentially non-zero) desired dissent.
    Thus, if the authority wants to achieve zero political cost (i.e., $\sum_{i=1}^n\action{i}{r} = 0$), it must necessarily set $\tolerate{}^* = 0$.
    Substituting this value along with perfect surveillance ($\surveil{}^* = 1$) into individuals' optimal actions (Theorems~\ref{thm:optaction:uniform} and~\ref{thm:optaction:proportional}) yields
    \begin{align}
        \left.\action{i}{r}^* \right|_{\tolerate{r} = 0, \surveil{r} = 1, \punish = \text{uniform}} &= \left\{\begin{array}{ll}
            \desire{i} & \text{if $\desire{i} > \severity{r} / \boldness{i}$}; \\
            0 & \text{otherwise.}
        \end{array}\right. \\
        \left.\action{i}{r}^* \right|_{\tolerate{r} = 0, \surveil{r} = 1, \punish = \text{proportional}} &= \left\{\begin{array}{ll}
            \desire{i} & \text{if $\boldness{i} \geq \severity{r}$}; \\
            0 & \text{otherwise.}
        \end{array}\right.
    \end{align}
    Thus, the authority causes every individual to self-censor ($\action{i}{r}^* = \tolerate{}^* = 0$) if and only if $\severity{r} \geq \desire{i} \cdot \boldness{i}$ under uniform punishment or $\severity{r} > \boldness{i}$ under proportional punishment, for all individuals $I_i$.
\end{proof}

This draconian policy is, in fact, a Nash equilibrium if authority--individual interactions are formulated as a one-round leader-follower game.
Specifically, if the authority acts first and individuals act second, without collusion and under the assumption that the authority's threat of punishment is credible then we have the following result.

\begin{corollary}
    The authority's draconian policy and the corresponding individuals' optimal actions form a Nash equilibrium.
\end{corollary}
\begin{proof}
    In the proof of Theorem~\ref{thm:draconian}, we showed that under the draconian policy $(\tolerate{}^*, \severity{}^*, \surveil{}^*)$, every individual completely self-censors $(\action{i}{r}^* = \tolerate{}^* = 0)$ regardless of punishment function.
    The uniqueness of this utility-maximizing action guaranteed by Theorems~\ref{thm:optaction:uniform} and~\ref{thm:optaction:proportional} ensures that no individual can deviate from this strategy of complete self-censorship without harming their utility.
    The proof of Theorem~\ref{thm:draconian} also shows that the authority's utility is nonpositive (i.e., $U_{A,r} \leq 0$) and that the draconian policy achieves the maximum possible value, $U_{A,r} = 0$.
    Thus, it is clear that no deviation from this strategy could improve the authority's utility.
    Therefore, these strategies form a (not necessarily unique) Nash equilibrium.
\end{proof}

\subsection{Adaptive Authority Experiments}

Algorithm~\ref{alg:rmhc} details pseudocode for the random mutation hill climbing authority's adaptation rules.

\begin{algorithm}[H]
\caption{Random Mutation Hill Climbing Authority: Single Trial} \label{alg:rmhc}
\begin{algorithmic}[1]
    \State \textbf{Inputs}:
    \begin{itemize}
        \item $n > 0$: an integer number of individuals in the population
        \item $R > 0$: an integer number of rounds to simulate
        \item $\desire{} \in (0, 0.5)$: a float mean population desired dissent
        \item $\boldness{} > 0$: a float mean population boldness
        \item $\punish$: the authority's punishment function
        \item $\tolerate{0} \in [0, 1]$: the authority's float initial tolerance
        \item $\severity{0} > 0$: the authority's float initial punishment severity
        \item $\surveil{0} \in [0, 1]$: the authority's float initial surveillance
        \item $\alpha > 0$: the authority's float adamancy
        \item $\varepsilon > 0$: the float radius of the random mutation update window
    \end{itemize}

    \State Sample $\{\desire{i}\}_{i=1}^n$ from a truncated exponential distribution with mean $\desire{}$ and bounds $[0, 1]$.

    \State Sample $\{\boldness{i}\}_{i=1}^n$ from an exponential distribution with mean $\boldness{}$.

    \State Let $\{\action{i}{0}^*\}_{i=1}^n$ be the individuals' optimal actions under the authority's initial policy $(\tolerate{0}, \severity{0}, \surveil{0})$.

    \State Let $c_{\text{pol},0} = \sum_{i=1}^n \action{i}{0}^*$ and $c_{\text{pun},0} = \sum_{i=1}^n \punish(\action{i}{0}^*, \tolerate{0}, \severity{0}, \surveil{0})$ be the authority's initial political and punishment costs.

    \For {$r \in \{1, \ldots, R\}$}
        \State Choose a parameter $p \in \{\tolerate{}, \severity{}, \surveil{}\}$ uniformly at random.

        \If {$p = \tolerate{}$}
            \State Choose $\tolerate{}'$ uniformly at random in $[\max\{0, \tolerate{r-1} - \varepsilon\}, \min\{1, \tolerate{r-1} + \varepsilon\}]$.
            \State Set $\severity{}' \gets \severity{r-1}$ and $\surveil{}' \gets \surveil{r-1}$.
        \ElsIf {$p = \severity{}$}
            \State Choose $\severity{}'$ uniformly at random in $[\max\{10^{-9}, \severity{r-1} - \varepsilon\}, \severity{r-1} + \varepsilon]$.
            \State Set $\tolerate{}' \gets \tolerate{r-1}$ and $\surveil{}' \gets \surveil{r-1}$.
        \ElsIf {$p = \surveil{}$}
            \State Choose $\surveil{}'$ uniformly at random in $[\max\{0, \surveil{r-1} - \varepsilon\}, \min\{1, \surveil{r-1} + \varepsilon\}]$.
            \State Set $\tolerate{}' \gets \tolerate{r-1}$ and $\severity{}' \gets \severity{r-1}$.
        \EndIf

        \State Let $\{\action{i}{r}^*\}_{i=1}^n$ be the individuals' optimal actions under the mutated policy $(\tolerate{}', \severity{}', \surveil{}')$.

        \State Let $c_{\text{pol},r} = \sum_{i=1}^n \action{i}{r}^*$ and $c_{\text{pun},r} = \sum_{i=1}^n \punish(\action{i}{r}^*, \tolerate{}', \severity{}', \surveil{}')$ be the authority's political and punishment costs in round $r$.

        \If {$\alpha \cdot c_{\text{pol},r} + c_{\text{pun},r} \leq \alpha \cdot c_{\text{pol},r-1} + c_{\text{pun},r-1}$}
            \State Set $(\tolerate{r}, \severity{r}, \surveil{r}) \gets (\tolerate{}', \severity{}', \surveil{}')$.
        \Else {}
            \State Set $(\tolerate{r}, \severity{r}, \surveil{r}) \gets (\tolerate{r-1}, \severity{r-1}, \surveil{r-1})$.
        \EndIf
    \EndFor

    \State \Return the authority's policies $\{(\tolerate{r}, \severity{r}, \surveil{r})\}_{r=0}^R$ and costs $\{(c_{\text{pol},r}, c_{\text{pun},r})\}_{r=0}^R$ and the individuals' parameters $\{(\desire{i}, \boldness{i})\}_{i=1}^n$.
\end{algorithmic}
\end{algorithm}

All experiments were run on a Linux machine with a 5.7 GHz Ryzen 9 7950X CPU (16 cores, 32 threads) and 64 GB of memory.
The time evolutions shown in \figtext~\ref{fig:adaptiveauthorityevo} and \figtext~\ref{si:fig:adaptiveauthorityevo} are obtained with:
\begin{verbatim}
  python hillclimbing.py -N 100000 -R 5000 -D 0.25 -B 2 -P uniform --seed 458734673

  python hillclimbing.py -N 100000 -R 5000 -D 0.25 -B 8 -P uniform --seed 458734673

  python hillclimbing.py -N 100000 -R 5000 -D 0.25 -B 1 -P proportional \
  --seed 458734673

  python hillclimbing.py -N 100000 -R 5000 -D 0.25 -B 3 -P proportional \
  --seed 458734673
\end{verbatim}
The parameter sweeps shown in \figtext~\ref{fig:adaptiveauthoritysweep} and \figstext~\ref{si:fig:adaptiveauthoritysweep}--\ref{si:fig:suppressiontime} are obtained with:
\begin{verbatim}
  python hillclimbing.py --sweep -N 100000 -R 10000 -P uniform -A 1.0 -E 0.05 \
  --seed 1978572166 --granularity 50 --trials 50 --threads 32

  python hillclimbing.py --sweep -N 100000 -R 10000 -P proportional -A 1.0 -E 0.05 \
  --seed 1978572166 --granularity 50 --trials 50 --threads 32
\end{verbatim}

\subsection{Preliminary Explorations of Assimilative Opinion Dynamics}

Suppose that the individuals $\mathcal{I}$ are arranged in a connected, undirected network $G = (\mathcal{I}, E)$ representing their population's interaction structure and some \textit{social influence} occurs among neighbors at the start of each round.
For example, ``socialization'' could occur by shifting each individual's desired dissent towards its neighbors' mean action in the last round, which then affects the individual's optimal action in this round according to Theorems~\ref{thm:optaction:uniform} or~\ref{thm:optaction:proportional}.
Among the plethora of possible opinion dynamics, only those that change individuals' desired dissents or boldness values are interesting.
Rules that only modify individuals' actions (e.g., ``enact the action halfway between your optimal action $\action{i}{r}^*$ and the mean of your neighbors' actions in the last round'') are inherently anchored to individuals' unchanging, underlying parameters, stopping the population from exhibiting any long-term deviations from their initial behavior.

Another class of interaction rules yielding a predictable outcome is \textit{distributed averaging}.
In each round, individuals update their parameters by shifting them some fraction of the way towards the (possibly weighted) average of their neighbors' parameters in the last round:
\begin{equation}
    \desire{i,r+1} = \frac{w \cdot \desire{i,r} + \sum_{I_j \in N(I_i)} \desire{j,r}}{w + |N(I_i)|}
    \quad\quad\text{or}\quad\quad
    \boldness{i,r+1} = \frac{w \cdot \boldness{i,r} + \sum_{I_j \in N(I_i)} \boldness{j,r}}{w + |N(I_i)|},
\end{equation}
where $N(I_i)$ is the neighborhood of individual $I_i$ in $G$ and $w > 0$ is a weight parameter.
Perhaps surprisingly, regardless of the structure of the network $G$ or the weight parameter $w$, these rules always yield consensus at the centrality-weighted average of the individuals' initial values~\cite{Como2016-localaveraging}.
Thus, the long-run dissent behavior of a networked population of individuals performing distributed averaging can be predicted directly from (\textit{i}) each individual's initial desired dissent and boldness values, (\textit{ii}) each individual's centrality in the network, and (\textit{iii}) and the authority's tolerance, severity, and surveillance parameters.

Supposing, as we do in our adaptive authority experiments, that a population's desired dissents and boldness values are exponentially distributed (i.e., very few individuals are initially high-desire, high-boldness), this result suggests that social influence almost always inflicts a chilling effect, pulling more extreme dissidents towards a self-censoring center instead of the other way around.
The only exception is when these more extreme dissidents are exponentially more central in the network than everyone else.
Thus, distributed averaging alone is likely insufficient for population-level cascades of dissent that we observe in reality~\cite{Chenoweth2011-whycivil,Chenoweth2021-civilresistance}.


\begin{figure}
    \centering
    \begin{subfigure}{.48\textwidth}
        \centering
        \caption{Uniform Punishment, $\boldness{} = 2$ (same as \figtext~\ref{fig:adaptiveauthorityevo})}
        \label{si:fig:adaptiveauthorityevo:uniformquell}
        \includegraphics[width=\textwidth]{rmhc_trial_N100000_R5000_uniform_B2_S458734673.pdf}
    \end{subfigure}
    \hfill
    \begin{subfigure}{.48\textwidth}
        \centering
        \caption{Uniform Punishment, $\boldness{} = 8$}
        \label{si:fig:adaptiveauthorityevo:uniformresist}
        \includegraphics[width=\textwidth]{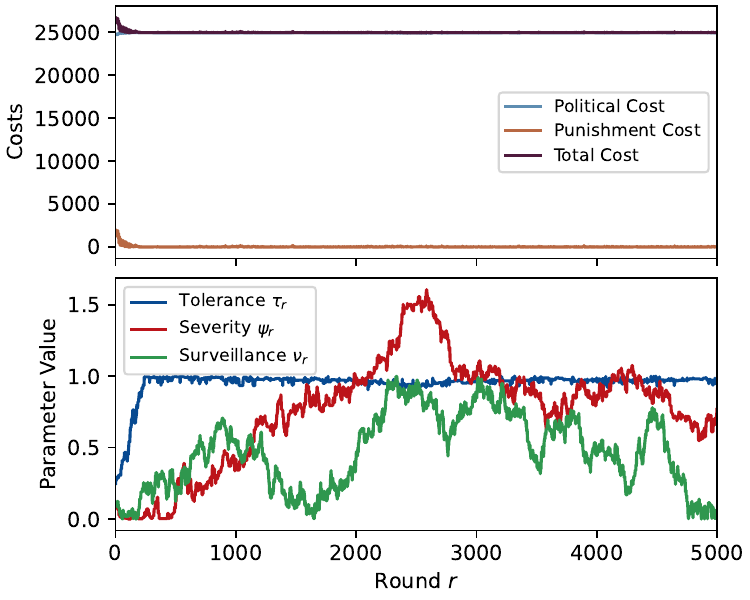}
    \end{subfigure}
    \\ \bigskip
    \begin{subfigure}{.48\textwidth}
        \centering
        \caption{Proportional Punishment, $\boldness{} = 1$}
        \label{si:fig:adaptiveauthorityevo:propquell}
        \includegraphics[width=\textwidth]{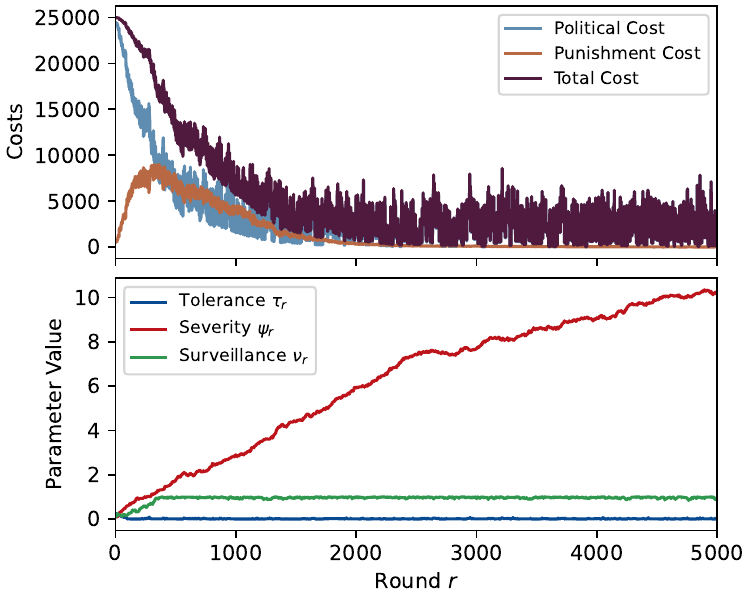}
    \end{subfigure}
    \hfill
    \begin{subfigure}{.48\textwidth}
        \centering
        \caption{Proportional Punishment, $\boldness{} = 3$}
        \label{si:fig:adaptiveauthorityevo:propresist}
        \includegraphics[width=\textwidth]{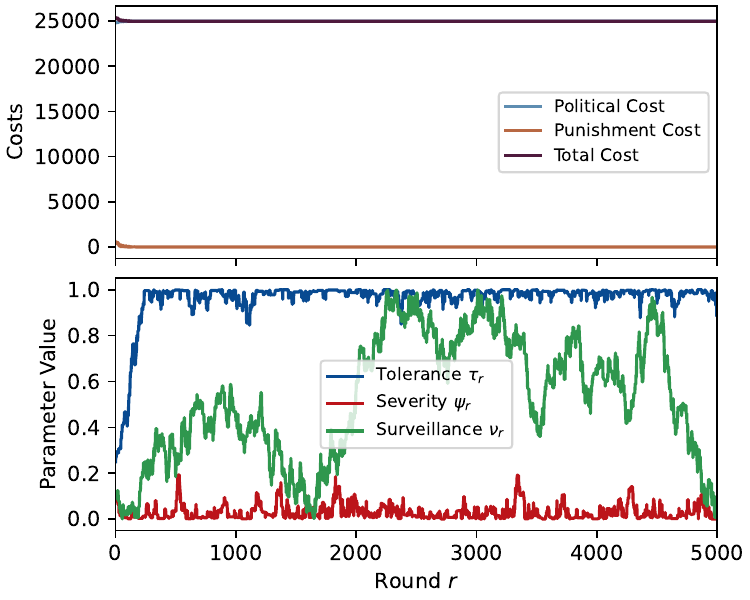}
    \end{subfigure}
    \caption{\textit{Time Evolutions of an Adaptive Authority.}
    Expanding \figtext~\ref{fig:adaptiveauthorityevo}, this shows the adaptive authority's costs (top panels) and parameters (bottom panels) over \numprint{5000} rounds with adamancy $\adamancy =$ 1 for different parameter regimes and punishment functions.
    In each trial, the population has $n =$ \numprint{100000} individuals whose desired dissents are sampled independently from a truncated exponential distribution with mean $\desire{} = 0.25$ and bounds $[0, 1]$.
    Their boldness constants are also sampled from exponential distributions whose means $\boldness{}$ are shown in the subcaptions.
    \textbf{(a)},\textbf{(c)} For populations with low boldness, regardless of punishment function, the adaptive authority rapidly discovers a near-draconian policy that suppresses most dissent (top panels, blue ``political cost'') without having to enact nearly any punishments (top panels, orange ``punishment cost'').
    \textbf{(b)},\textbf{(d)} For populations with sufficiently large boldness, the authority does not discover a dissent-suppressing policy within the simulated time frame, instead continuously exploring different tolerant policies ($\tolerate{} \approx 1$) where individuals fully express their desired dissent without fear of punishment.}
    \label{si:fig:adaptiveauthorityevo}
\end{figure}

\begin{figure}
    \centering
    \begin{subfigure}{.48\textwidth}
        \centering
        \caption{Uniform Punishment (same as \figtext~\ref{fig:adaptiveauthoritysweep})}
        \label{si:fig:adaptiveauthoritysweep:uniform}
        \includegraphics[width=\textwidth]{sweep_N100000_R10000_uniform_S1978572166.png}
    \end{subfigure}
    \hfill
    \begin{subfigure}{.48\textwidth}
        \centering
        \caption{Proportional Punishment}
        \label{si:fig:adaptiveauthoritysweep:prop}
        \includegraphics[width=\textwidth]{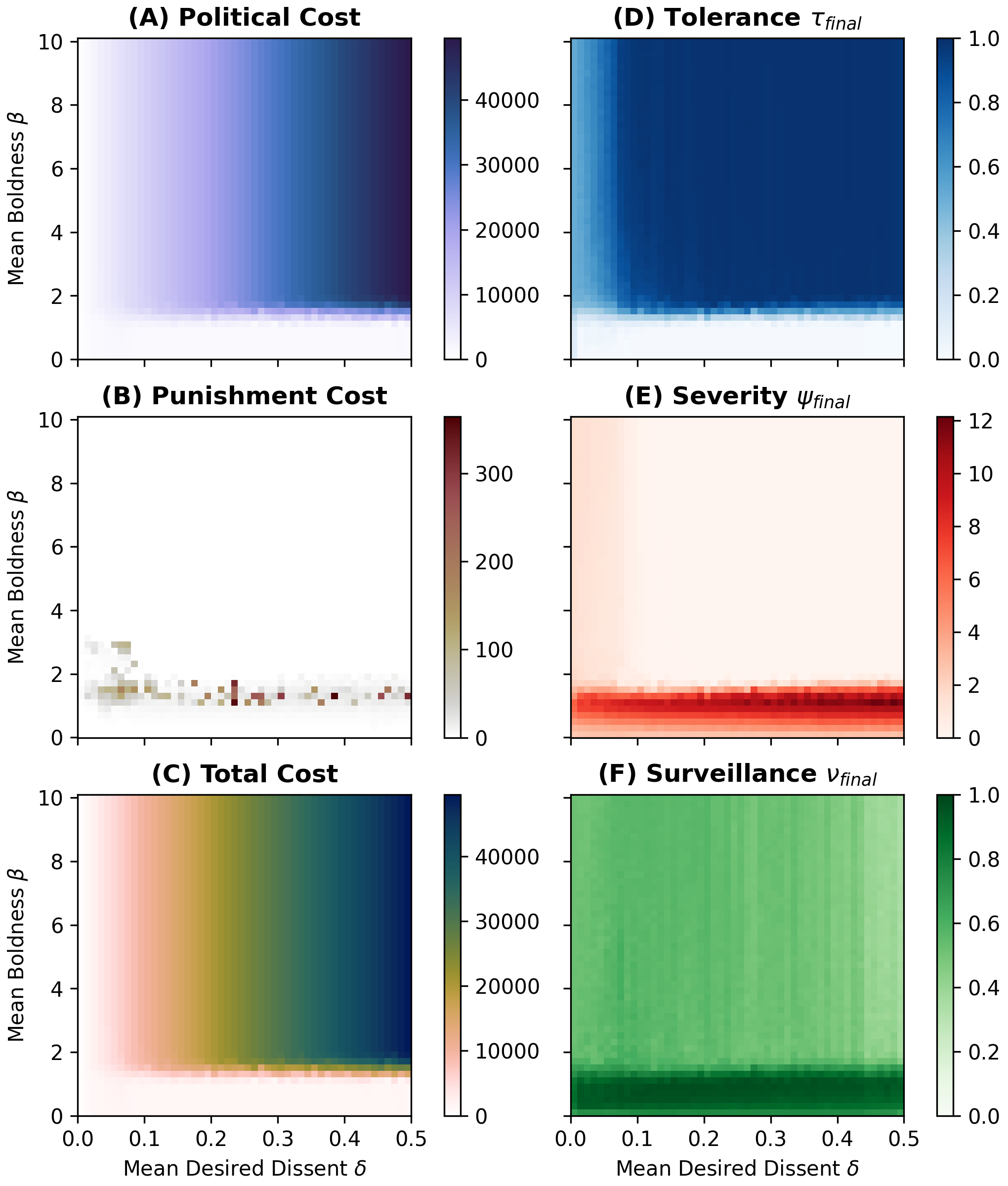}
    \end{subfigure}
    \caption{\textit{Authority Costs and Policies After Adaptation.}
    Expanding \figtext~\ref{fig:adaptiveauthoritysweep}, this shows the authority's average final \textbf{(A)}--\textbf{(C)} costs and \textbf{(D)}--\textbf{(F)} parameter values after \numprint{10000} rounds of random mutation hill climbing adaptation across 50 independent trials per $(\desire{}, \boldness{})$ pair, where $\desire{} \in [0.005, 0.495]$ and $\boldness{} \in [0.1, 10]$ are the means of the population's exponential distributions of desired dissents and boldness constants, respectively.
    Each trial involves $n =$ \numprint{100000} individuals and the authority uses adamancy $\adamancy =$ 1 with \textbf{(a)} uniform or \textbf{(b)} proportional punishment.
    Comparing panel \textbf{(C)} across punishment functions, adaptation with uniform punishment yields minimal total costs (i.e., near-total suppression of dissent and low punishment) for much larger regions of population parameter space than with proportional punishment.
    Under proportional punishment, a critical boldness value $\boldness{}^* \approx 1.2$ separates the discovery of a draconian policy ($\boldness{} < \boldness{}^*$, \figtext~\ref{si:fig:adaptiveauthorityevo:propquell}) from continuously exploring tolerant policies where individuals fully express their desired dissent ($\boldness{} > \boldness{}^*$, \figtext~\ref{si:fig:adaptiveauthorityevo:propresist}) regardless of the population's mean desired dissent $\desire{}$.}
    \label{si:fig:adaptiveauthoritysweep}
\end{figure}

\begin{figure}
    \centering
    \begin{subfigure}{.48\textwidth}
        \centering
        \caption{Uniform Punishment}
        \label{si:fig:suppressiontime:uniform}
        \includegraphics[width=\textwidth]{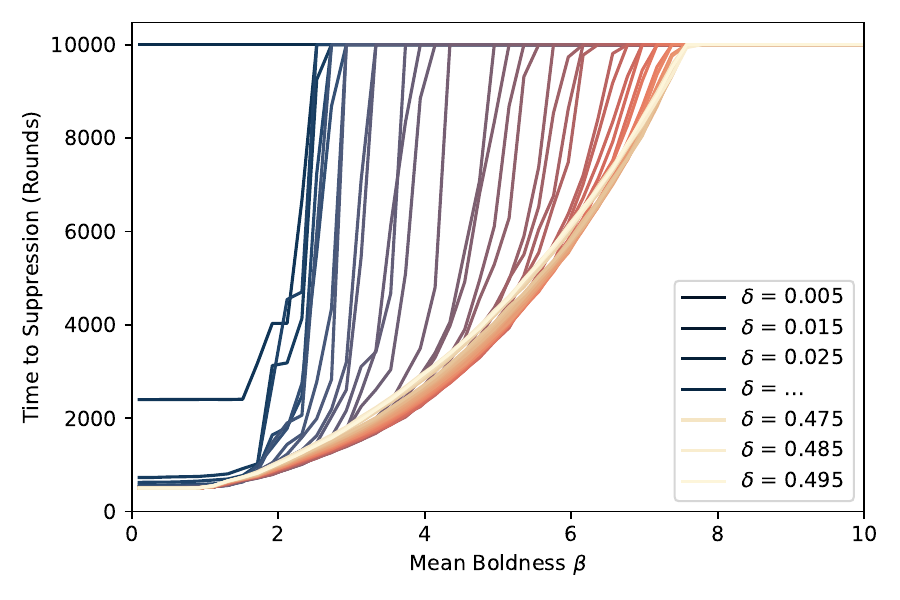}
    \end{subfigure}
    \hfill
    \begin{subfigure}{.48\textwidth}
        \centering
        \caption{Proportional Punishment}
        \label{si:fig:suppressiontime:prop}
        \includegraphics[width=\textwidth]{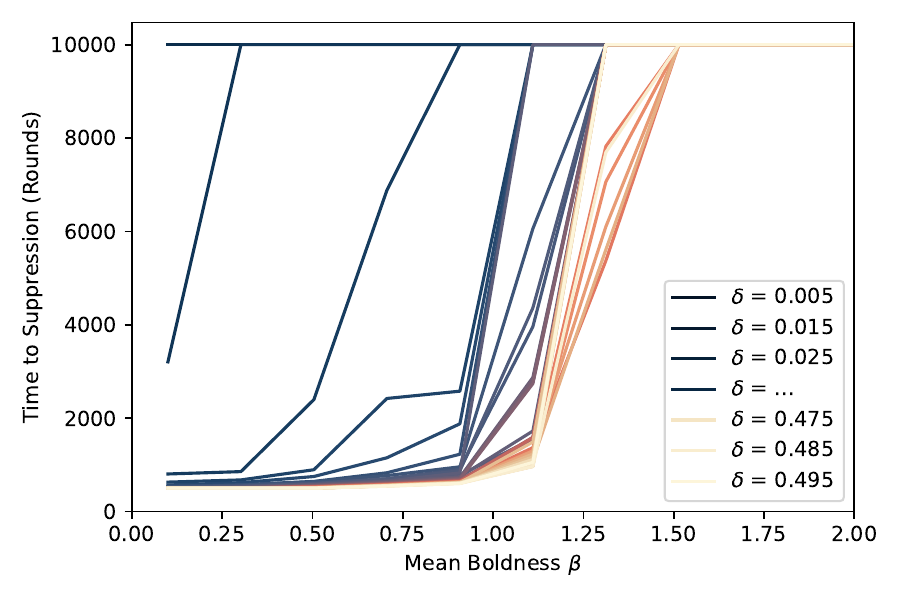}
    \end{subfigure}
    \caption{\textit{Times to Suppression by an Adaptive Authority.}
    We define the \textit{$(\gamma, w)$-time to suppression} by an adaptive authority as the earliest round $r \geq w$ such that the authority's mean political cost in a sliding window of rounds $[r-w, r]$ falls below a $\gamma$-fraction of the population's maximum possible political cost (i.e., $\sum_{i=1}^n\desire{i}$).
    Here, we show the (0.25, 500)-times to suppression for the adaptive authority trials shown in \figtext~\ref{si:fig:adaptiveauthoritysweep} under \textbf{(a)} uniform and \textbf{(b)} proportional punishment as a function of population mean boldness $\boldness{}$ for different population mean desired dissents $\desire{}$.
    In words, these are the amounts of time passed before the adaptive authority discovers a policy that decreases the population's mean total dissent in 500-round sliding windows to 25\% of its maximum value at full defiance.
    The maximum time to suppression of \numprint{10000} rounds shown here is simply the time horizon of our simulation runs.
    \textbf{(a)} Under uniform punishment, boldness has a superlinear effect on time to suppression, with larger mean population desired dissents $\desire{}$ requiring larger mean population boldness values $\boldness{}$ to sustain widespread expression of dissent for long durations.
    \textbf{(b)} Under proportional punishment, the transition from immediate suppression to maximum time to suppression is sharper, as already evidenced by the likely critical boldness value $\boldness{}^* \approx 1.2$ discussed in \figtext~\ref{si:fig:adaptiveauthoritysweep:prop}(C).}
    \label{si:fig:suppressiontime}
\end{figure}

\end{document}